\documentclass[runningheads,orivec,envcountsame]{llncs}

\usepackage{centernot}
\usepackage{amssymb}
\usepackage[linesnumbered,noend,boxruled]{algorithm2e}
\usepackage{stmaryrd}
\usepackage{bm}
\usepackage{tikz}
\usepackage{hhline}
\usepackage{pgfplots}
\usetikzlibrary{shapes,calc,arrows,automata,arrows.meta}
\usepackage{multirow}
\usepackage{array}
\usepackage{mathtools}
\usepackage{enumitem}
\usepackage[bookmarks,unicode,colorlinks=true]{hyperref}
\usepackage[capitalize,nameinlink]{cleveref}
\usepackage[location=appendix,manual]{moveproofs}
\usepackage{thm-restate}
\usepackage{cite}
\usepackage{listings}
\usepackage[dvipsnames]{xcolor}
\usepackage{wrapfig}
\usepackage{graphicx}
\usepackage{marvosym}

\usepackage{microtype}
\everypar{\looseness=-1}
\allowdisplaybreaks[4]
\everydisplay=\expandafter{\the\everydisplay\small
\setlength{\abovedisplayskip}{3pt}
\setlength{\belowdisplayskip}{3pt}
\setlength{\abovedisplayshortskip}{1pt}
\setlength{\belowdisplayshortskip}{1pt}}
\setlist{itemsep=1pt, topsep=1pt}

\DontPrintSemicolon

\makeatletter
\RequirePackage[bookmarks,unicode,colorlinks=true]{hyperref}%
   \def\@citecolor{blue}%
   \def\@urlcolor{blue}%
   \def\@linkcolor{blue}%

\def\orcidID#1{\href{http://orcid.org/#1}{\protect\raisebox{-1.25pt}{\protect\includegraphics{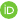}}}}
\makeatother

\makeatletter
\newcommand{\algorithmstyle}[1]{\renewcommand{\algocf@style}{#1}}
\newcommand{\nosemic}{\renewcommand{\@endalgocfline}{\relax}}
\newcommand{\dosemic}{\renewcommand{\@endalgocfline}{\algocf@endline}}
\let\oldnl\nl
\newcommand{\nonl}{\renewcommand{\nl}{\let\nl\oldnl}}
\makeatother

\hypersetup{
  pdftitle={Accelerating Loops with Arrays},
  pdfauthor={Florian Frohn and Jürgen Giesl},
  colorlinks=true,
  linkcolor=blue,
  citecolor=olive,
  filecolor=magenta,
  urlcolor=cyan
}

\usepackage[most]{tcolorbox}

\tcbset{on line, boxrule=0.5pt,
        boxsep=2pt, left=0pt,right=0pt,top=0pt,bottom=0pt,
        colframe=gray,colback=gray!5,coltext=blue,  
        highlight math style={enhanced}
        }

\newtagform{box}[\tcbox]{}{}
\usetagform{box}

\renewcommand{\eqref}[1]{\tcbox[boxsep=1.5pt,boxrule=0.5pt]{\small\ref{#1}}}
\newcommand{\scriptsizeeqref}[1]{\tcbox[boxsep=1pt,boxrule=0.3pt]{\scriptsize\ref{#1}}}

\renewcommand{\top}{\mathsf{true}}
\renewcommand{\bot}{\mathsf{false}}
 \newcommand{\twodots}{\mathinner {\ldotp \ldotp}}

\newcommand{\tool}[1]{{\sf #1}}

\renewcommand{\div}{\mathrel{\mathsf{div}}}
\newcommand{\divides}{\mathrel{\mid}}
\newcommand{\notdivides}{\mathrel{\nmid}}
\renewcommand{\epsilon}{\varepsilon}
\let\oldphi\phi
\let\oldvarphi\varphi
\renewcommand{\varphi}{\oldphi}
\renewcommand{\phi}{\oldvarphi}
\renewcommand{\emptyset}{\varnothing}

\newcommand{\rvec}[1]{\mathbf{#1}}

\newcommand{\ZZ}{\mathbb{Z}}

\newcommand{\NN}{\mathbb{N}}

\newcommand{\VV}{\mathcal{V}}
\newcommand{\LL}{\mathcal{L}}

\newcommand{\TT}{\mathcal{T}}

\newcommand{\RR}{\mathcal{R}}

\newcommand{\Def}{\mathrel{\mathop:}=}

\newcommand{\mat}[1]{\left(\begin{smallmatrix} #1 \end{smallmatrix}\right)}

\newcommand{\mDo}{\mathbf{do}}
\newcommand{\mElse}{\mathbf{else}}
\newcommand{\mDone}{\mathbf{done}}
\newcommand{\mWhile}[2]{\mathbf{while}\ #1\ \mDo\ #2\ \mDone}
\newcommand{\mIf}[2]{\mathbf{if}\ (#1)\ #2}
\newcommand{\mIte}[3]{\mathbf{if}\ (#1)\ #2\ \mElse\ #3}
\newcommand{\assign}{\leftarrow}

\newcommand{\Lval}{\mathit{Lval}}
\newcommand{\Rval}{\mathit{Rval}}

\newcommand{\Expr}{\mathit{Expr}}

\newcommand{\arity}{\mathsf{arity}}

\newcommand{\rec}{\mathsf{rec}}

\newcommand{\Rec}{\mathsf{Rec}}

\newcommand{\sem}[1]{\left\llbracket #1 \right\rrbracket}

\newcommand{\lam}[4]{\lambda #1.\ \mIte{#1 = #2}{#3}{#4}}
\newcommand{\up}{\mathsf{up}}
\newcommand{\rhs}{\mathsf{rhs}}
\newcommand{\written}{\mathsf{written}}
\newcommand{\notwritten}{\mathsf{!written}} 
\newcommand{\lastwrite}{\mathsf{last\_write}}

\newcommand{\unsat}{\mathsf{unsat}}
\newcommand{\unknown}{\mathsf{unknown}}

\SetCommentSty{mycommfont}

\newcommand{\oldcomment}[1]{}

\DeclareMathOperator{\dom}{dom}
\DeclareMathOperator{\img}{img}
\DeclareMathOperator{\true}{\top}
\DeclareMathOperator{\false}{\bot}

\crefname{equation}{eq.}{equations}%
\crefname{chapter}{chapter}{chapters}%
\crefname{section}{sect.}{sections}%
\crefname{appendix}{app.}{appendices}%
\crefname{enumi}{item}{items}%
\crefname{footnote}{footnote}{footnotes}%
\crefname{figure}{fig.}{figures}%
\crefname{table}{table}{tables}%
\crefname{theorem}{thm.}{theorems}%
\crefname{lemma}{lemma}{lemmas}%
\crefname{corollary}{cor.}{corollaries}%
\crefname{proposition}{proposition}{propositions}%
\crefname{definition}{def.}{defs.}%
\crefname{result}{result}{results}%
\crefname{example}{ex.}{examples}%
\crefname{remark}{remark}{remarks}%
\crefname{note}{note}{notes}%

\title{Accelerating Loops with Arrays}
\author{Florian Frohn$^{(\href{mailto:florian.frohn@informatik.rwth-aachen.de}{\mbox{\Letter}})}$\orcidID{0000-0003-0902-1994} \and Jürgen Giesl$^{(\href{mailto:giesl@informatik.rwth-aachen.de}{\mbox{\Letter}})}$\orcidID{0000-0003-0283-8520}}
\institute{RWTH Aachen University, Aachen, Germany}
    \authorrunning{F.\ Frohn, J.\ Giesl}

\begin{document}

\maketitle

\begin{abstract}
  We propose a novel
  \emph{acceleration technique} for
  loops operating on arrays.
  The goal of acceleration is to characterize the transitive closure of loops in a logic which is suitable for automated reasoning.
  Using the new notion of \emph{inductive lvalues}, our technique can handle
loops where previous techniques fail, and it \emph{unifies} acceleration for arrays and
scalar variables by regarding scalars as arrays of dimension 0. 
Moreover, our approach uses
 $\lambda$s instead of quantifiers.
  Then the resulting SMT problems can be solved via
  \emph{lemmas on demand}.
  An empirical evaluation of our implementation in the
  tool \textsf{LoAT}
  shows the power of our approach.
\end{abstract}

\section{Introduction}
\label{sec:intro}

Model checking programs operating on arrays is notoriously difficult.
One reason is that interesting properties of array programs often involve quantifiers, but reasoning about quantifiers is still challenging for SMT solvers.
In particular, they can hardly find models for such formulas, but models are crucial for 
model checking algorithms that use them to prove
unsafety and for abstraction refinement.

Alternatives to abstraction refinement are based on \emph{acceleration techniques}~\cite{bozga10,fast,Boigelot03,bozga09a,octagonsP,differenceBounds,acceleration-calculus,solvable} that try to characterize the transitive closure of a single-path loop.
Then algorithms like Acceleration Driven Clause Learning \cite{adcl} or Accelerated Bounded Model Checking (ABMC) \cite{abmc} can lift acceleration from single-path loops to full programs.
As acceleration does not over-approximate, it is particularly useful for finding counterexamples, but it can also be used for proving safety~\cite{flata}.

We present a new acceleration technique for loops operating on arrays.
In contrast to earlier approaches
\cite{aligators,underapprox15,alberti13,alberti15,albertiJAR15}, it applies to a broader class of loops, due to the novel notion of \emph{inductive lvalues}.
The core idea of inductive lvalues is that the effect of assignments like
$a[i+1] \gets
a[i] + 1; i \gets i+1$ in loop bodies
can be summarized by recurrence equations.
Here, the value $a[i]$ that is read is equal to the value that has been written to $a[i+1]$ in the
previous iteration.

Another benefit of inductive lvalues is that they allow us to \emph{unify} loop acceleration for arrays and scalar variables by regarding scalar variables as arrays of dimension $0$.
In contrast, earlier approaches treat arrays and scalars differently.

One more key difference to earlier approaches is that we
use
$\lambda$s and a specialized quantifier elimination technique to avoid quantifiers.
Even though most SMT solvers have very limited support for $\lambda$s, the resulting SMT problems can usually be solved with a refinement loop in the spirit of \emph{lemmas on demand}~\cite{lemmas-on-demand,lemmas-phd,DBLP:conf/fmcad/PreinerNB13}.

\section{Overview and Contributions}
\label{sec:overview}

\begin{wrapfigure}{r}{0.3\textwidth}
  \vspace{-0.2em}
\begin{lstlisting}[language=C,basicstyle=\scriptsize,frame=single]
int i = 0;
int k = 10000;
while(i < k) {
    a[i+1] = a[i];
    i = i + 1;
}
int j = nondet(0, k);
assert(a[j] != a[0]);
\end{lstlisting}
\end{wrapfigure}
\noindent
We start with an informal overview, where we illustrate our pipeline for verifying array programs, and explain which parts are contributions of this paper.
Consider the {\sf C} program on the right, which writes $a[0]$ to $a[1],\ldots,a[k]$.
Here, \lstinline{nondet(0, k)} returns some integer between $0$ and
$k$. So the final assertion always fails.
There exist various techniques to transform imperative programs into more formal
representations like \emph{transition systems}, resulting in:
\begin{align}
  \mathsf{init}(a) & {} \to \mathsf{while}(a,0,10000) \notag \\
  \mathsf{while}(a,i,k) & {} \to \mathsf{while}(\lambda j.\ \mathrlap{\mIte{j=i+1}{a[i]}{a[j]},i+1,k)} \hspace{10em} & \mathbf{if}\ i < k \label{eq:loop} \tag{\sc Loop}\\
  \mathsf{while}(a,i,k) & {} \to \mathsf{check}(a,j) & \mathbf{if}\ i \geq k \land 0 \leq j \leq k \notag\\
  \mathsf{check}(a,j) & {} \to \mathsf{err} & \mathbf{if}\ a[j] = a[0] \notag
\end{align}
In our evaluation (see \Cref{sec:eval}), this step is done by our tool {\sf
HornKlaus} that transforms \textsf{C} programs into linear Constrained Horn Clauses, which
are equivalent to transition systems.
In the transition \eqref{eq:loop}, ``$\lambda j.\ \mIte{j=i+1}{a[i]}{a[j]}$'' denotes the array that results from writing $a[i]$ to $a[i+1]$.

To prove unsafety, one has to show that the location $\mathsf{err}$ can be reached from $\mathsf{init}$.
The core idea of acceleration is to replace loops like \eqref{eq:loop} with transitions like the one below that can simulate arbitrarily many iterations in a single step:
\begin{align}
  \mathsf{while}(a,i,k) \to \mathsf{while}(a',i+n,k) \ \ \ \mathbf{if}\ & n > 0 \land i + n \leq k \land \forall j \in \ZZ \label{eq:accel} \tag{\sc Accel}\\ 
                                                                         & \quad (i < j \leq i+n \implies a'[j] = a[i]) \land {} \notag \\
                                                                         & \quad (j \leq i \lor i+n < j \implies a'[j] = a[j]) \notag
\end{align}
It simulates $n$ iterations, where $n$ is chosen non-deterministically.

For this step, two ingredients are needed:
First, one needs a \emph{closed form} for each variable $x$ that expresses the value of $x$ after $n$ iterations.
For scalar variables like $i$, \emph{closed form expressions} are
computed by solving recurrence equations, resulting in the closed form $i+n$.
So a closed form expression for a variable $x$ involves program variables and the additional variable $n$, and its value is equal to the value of $x$
after $n$ iterations, for all instantiations of the occurring
variables.

For arrays, existing acceleration techniques characterize closed
forms
via quantified formulas of the form ``$\forall j \in \ZZ.\ \phi$'' instead of expressions, where $j$ is used
as array index, as in \eqref{eq:accel}.
The reason is that the standard theory of arrays
\cite{mccarthy-arrays}
is not expressive enough to admit suitable
closed form expressions.

In this paper, we propose to use $\lambda$-expressions as a suitable language to express closed forms for arrays.
In our example, instead of \eqref{eq:accel}, our approach~yields:
\begin{equation}
  \label{eq:lambda}
  \tag{$\lambda$}
  \begin{aligned}
    &\mathsf{while}(a,i,k) \to {} \\
    &\mathsf{while'}\left(\lambda j.\ \mIte{i < j \leq i+n}{a[i]}{a[j]} ,i+n,k\right)
  \end{aligned}
  \ \mathbf{if}\ n > 0 \land i + n \leq k.
\end{equation}
\noindent
Moreover, our approach can characterize closed forms exactly in cases where earlier approaches fail or approximate.
To achieve this, we introduce \emph{inductive lvalues}.
An lvalue is an expression like $i$ or $a[i]$ that refers to a memory location.
We call an lvalue like $a[i]$ in \eqref{eq:loop} \emph{inductive}, as the value that is read from $a[i]$ is equal to the value that has been written to $a[i+1]$ in the previous iteration.
Thus, the value that is referenced by $a[i]$ in the $n^{th}$ iteration can be defined recursively, which gives rise to recurrence equations, as in the case of scalars. 
This allows us to accelerate loops like \eqref{eq:loop}
where earlier approaches fail, and to handle arrays and scalars uniformly by regarding scalars as arrays of dimension~$0$.

While $\lambda$s can replace the outermost universal quantifier (as in the step from \eqref{eq:accel} to \eqref{eq:lambda}), in general, nested quantifiers are needed to characterize closed forms for arrays.
As most acceleration based verification techniques rely on SMT solvers, and SMT solvers still struggle with quantifiers, our goal is to also eliminate the remaining quantifiers.
To this end, we identify an important class of loops where all quantifiers can indeed be
eliminated, see \Cref{sect:qe}.

The second ingredient for acceleration is a formula which guarantees that the loop can be executed $n$ times.
In our example, this is ensured by the constraint $i+n \leq k$.
In our implementation, such formulas are synthesized using the technique from
\cite{acceleration-calculus}.
The details are beyond the scope of this paper, but one core idea is to
exploit monotonicity:
In our example, if $i<k$ holds after executing the loop body, then it also held
before. So if $i<k$ holds before the $n^{th}$ iteration, then the loop can be executed $n$ times. At this point,
the value of $i$ is $i+n-1$, so we obtain $i+n \leq k$. 
While originally developed for integer loops, the approach from \cite{acceleration-calculus} immediately carries over to our setting.

Now our verification problem can be encoded into an SMT formula with $\lambda$s:
\begin{align*}
  &{\overbrace{i = 0 \land k = 10000}^{\mathclap{\text{initialization}}}} \land {\overbrace{n > 0 \land i + n \leq k}^{\text{guard of \eqref{eq:lambda}}}} \land {} \\
  &i' = i + n \land a' = \lambda j.\ \mIte{i < j \leq i+n}{a[i]}{a[j]} \land {\quad\Big\}} \text{update of \eqref{eq:lambda}} \\
  &{\underbrace{i' \geq k \land 0 \leq j \leq k}_{\text{transition to } \mathsf{check}}} \land {\underbrace{a'[j] = a'[0]}_{\mathclap{\text{transition to } \mathsf{err}}}}
\end{align*}
Unfortunately, the major SMT solvers usually fail if $\lambda$s cannot be $\beta$-reduced (i.e., evaluated).
However, SMT problems with $\lambda$s can often be solved via
\emph{lemmas on demand}~\cite{lemmas-on-demand,lemmas-phd,DBLP:conf/fmcad/PreinerNB13}.
Here, the main idea is to abstract $\lambda$s with uninterpreted functions, and refine this abstraction with suitable \emph{lemmas} whenever a model for the abstraction cannot be lifted to the original problem.
In this way, we can prove satisfiability of the formula above, which implies unsafety of our example.

To summarize, our main contributions are:
{\bf(1)} \emph{inductive lvalues} that allow for accelerating loops where existing techniques fail,
{\bf(2)} the first \emph{uniform} approach for\linebreak handling scalars and arrays in acceleration,
{\bf(3)} the use of $\lambda$s to accelerate arrays, where the resulting SMT problems can be solved via \emph{lemmas on demand},
{\bf(4)} a quantifier elimination technique which often yields quantifier-free closed
forms (with $\lambda$s) for arrays, and
{\bf(5)} a competitive implementation in our tool \tool{LoAT}.
\subsubsection{Outline:}
  After introducing preliminaries in \Cref{sec:preliminaries}, our main contributions are presented in \Cref{sec:closed,sect:qe}.
Related work is discussed in \Cref{sec:related}, and we evaluate our approach empirically and conclude in \Cref{sec:eval}.
App.~\ref{sec:lambdas} explains how we handle $\lambda$s in SMT, and App.~\ref{sec:proofs} contains all proofs.

\section{Preliminaries}
\label{sec:preliminaries}

We write $\vec{v}$ and $\rvec{v}$ for column and row vectors, and $v_i$
is the $i^{th}$ element of $\vec{v}$ or $\rvec{v}$.

\subsubsection{Arrays and Expressions:}

Let $\VV$ be a countably infinite set of variables,
where each $x \in \VV$ has a designated \emph{arity} $\arity(x) \in \NN$.
Then $x$ ranges over $\ZZ^{\ZZ^{\arity(x)}}$, i.e., over all functions with domain
$\ZZ^{\arity(x)}$ and codomain $\ZZ$ (i.e., $x$ maps row vectors to integers).
Intuitively, $x$ represents an integer array of dimension $\arity(x)$, i.e., we only consider integer valued arrays.
So $x$ represents a scalar if $\arity(x) = 0$, and a 1-dimensional array which is infinite in both directions if $\arity(x) = 1$.
Considering infinite arrays is in line with SMT-LIB~\cite{SMTLIB}.
To model finite arrays, array-accesses must be guarded with
suitable constraints.
Note that we interpret arrays as uncurried functions, i.e., $x[\rvec{i}]$ is not
a valid expression if $|\rvec{i}| \neq \arity(x)$.
So if $a$ is a 2-dimensional array, then $a[0]$ is undefined.
In contrast, $a[0]$ would be a 1-dimensional array in, e.g., {\sf C}.
Handling ``curried arrays'' is future work.

We use $c,i,j,k,m,n$ for scalar variables and integer constants, $a,b$ for
non-scalars (i.e., for arrays of dimension $> 0$), and $x,y$ can represent both scalars and non-scalars.
For any entity $e$, $\VV(e)$ denotes the variables that occur freely in $e$.

An expression of the form $x[\rvec{c}]$ where $\rvec{c} \in \ZZ^{\arity(x)}$
refers to a  \emph{cell}, where
each cell stores an integer value.
The set $\Rval$ of \emph{rvalues} is defined by the grammar
\begin{align*}
  \Rval \mathrel{\mathop{::=}} c \mid \Rval \circ \Rval \mid x[\rvec{r}]
\end{align*}
where $c \in \ZZ$, ${\circ}$ is an arithmetic operator,\footnote{We use ${+},{-},{\cdot},$ and ${\div}$ (for integer division), but other operators are possible as well.} $x \in \VV$, and $\rvec{r} \in \Rval^{\arity(x)}$.
The set $\Lval$ of \emph{lvalues} is defined as
\begin{align*}
  \Lval \Def \{x[\rvec{r}] \mid x \in \VV, \rvec{r} \in \Rval^{\arity(x)}\}.
\end{align*}
So an rvalue (lvalue) is an expression that may occur on the \underline{r}ight-hand (\underline{l}eft-hand) side
of an assignment.
We write $i$ instead of $i[\,]$ if $i \in \VV$ has arity $0$. 
For rvalues $r \in \Rval$ we define $\Lval(r)$ to be the set of $r$'s ``top-level'' lvalues, i.e.,
\[
  \Lval(r) \Def
  \begin{cases}
    \{r\} & \text{if } r \in \Lval \\
    \Lval(r_1) \cup \Lval(r_2) & \text{if } r = r_1 \circ r_2\\
    \emptyset & \text{if } r \in \ZZ
  \end{cases}
\]
So we have $\Lval(a[i] + 7) = \{a[i]\}$ and, importantly, $i \notin \Lval(a[i] + 7)$, i.e., $\Lval(r)$ only contains the lvalues occurring in $r$ that are not nested below other lvalues.

An \emph{array expression} is a variable from $\VV$, or of the form $\lambda
\rvec{i}.\ e$,
where $\rvec{i}$ is a vector of scalar variables, and $e \in \Expr$ is an \emph{expression}.
We define $\arity(\lambda \rvec{i} \ldots) \Def |\rvec{i}|$.
Expressions are generated by the grammar
\begin{align*}
  \Expr \mathrel{\mathop{::=}} c \mid \Expr \circ \Expr \mid p[\rvec{e}] \mid \mIte{\mu}{\Expr}{\Expr}
\end{align*}
where $c \in \ZZ$, ${\circ}$ is again an arithmetic operator, $p$ is an array expression, $\rvec{e} \in \Expr^{\arity(p)}$, and $\mu$ is a first-order formula over expressions involving the relations ${>},{\neq},\ldots$, and ${\divides}$, and the usual Boolean connectives.
Here, ${\divides}$ denotes divisibility, i.e., $c \divides i$ is true iff $c$ divides $i$.
So ``$\mIte{j=i}{a[i-1]}{a[j]}$'' is an expression, and ``$\lambda
j.\ \mIte{j=i}{a[i-1]}{a[j]}$'' is an array expression of arity 1.
We use $r$ for rvalues, $\ell$ for lvalues, $p,q$ for array expressions, and $e$ for expressions.

\subsubsection{Loops:}

We consider single-path loops (without branching in their body):
\[
  \label{loop}
  \tag{\ensuremath{\TT_{loop}}}
  \mWhile{\phi}{
      \vec{\ell}
    \assign
      \vec{r}
  }
\]
Here, $\phi$ is a conjunction of (in)equations over rvalues, $\vec{\ell} = \mat{x_1[\rvec{r}_1]\\ \cdots\\x_{m}[\rvec{r}_{m}]}$ is a vector of lvalues such that\footnote{\label{fn:guard}A less restrictive definition would require $\phi \implies \eqref{eq:distinct-indices}$.
  We use the definition above instead to simplify the presentation.
}
\begin{equation}
  \label{eq:distinct-indices}
  \tag{\sc Distinct}
  \forall i,j \in [1 \twodots m].\ i \neq j \implies x_i \neq x_j \lor \rvec{r}_i \neq \rvec{r}_j
\end{equation}
holds, and $\vec{r}  \in \Rval^{m}$.
Here, $[1 \twodots m]$ denotes the closed integer interval $\{1,\ldots,m\}$ (and ``$[1 \twodots m)$'' denotes a half-open integer interval).
Intuitively, the body of \ref{loop} updates all lvalues $\vec{\ell}$ simultaneously\footnote{\label{fn:seq}Considering sequential updates instead requires additional technicalities, so we restrict our attention to simultaneous updates for the sake of simplicity.} by overwriting each $x_i[\rvec{r}_i] \in \vec{\ell}$ with $r_i$, which is the $i^{th}$ component of $\rvec{r}$.
We also use the notation $\rhs(x_i[\rvec{r}_i]) \Def r_i$.
Such loops correspond to recursive transitions like \eqref{eq:loop} in \Cref{sec:overview}.
\eqref{eq:distinct-indices} ensures that two different lvalues cannot
refer to the same cell, which prevents loops like
\[
  \mWhile{a[j] = 0}{
    \mat{
      a[i]\\
      a[j]
    }
    \gets
    \mat{
      0\\
      1
    }
  }.
\]
As $a[i]$ and $a[j]$ are written simultaneously, here it is unclear if $a[j]$ is $0$ or $1$ after the loop body if $i=j$.
In the sequel, \ref{loop} denotes an arbitrary but fixed loop.

\begin{example}[Running Example]
  \label{ex:swap}
  We use the following loop as running example:
  \begin{equation}
    \label{leadLoop}
    \tag{\protect{\ensuremath{\TT_{\mathit{swap}}}}}
    \mWhile{i < k}{
      \mat{i \\ a[i+1] \\ a[i]} \assign \mat{i+1 \\ a[i] \\ a[i+1]}
    }
  \end{equation}
  It swaps the values of $a[i]$ and $a[i+1]$ in each iteration.
  So if the initial value of $a[i]$ is $c$, then after $n$ iterations the value of
  $a[i+n]$ is $c$, and the initial values of $a[i+1],\ldots,a[i+n]$ 
  have been shifted one position to the left.
\end{example}
For each $x \in \VV$,
$\up_x(\vec{\ell})$ is an array expression that describes the array  $x$ after one iteration of the
loop body. More precisely, we define $\up_x(\vec{\ell})$ inductively as follows:
  \begin{equation}
  \tag{$\begin{smallmatrix} \mathit{read}\\\mathit{write} \end{smallmatrix}$}
  \label{eq:sem}
  \hspace*{-2.2cm}
  \up_x(\vec{\ell}) \Def \begin{cases}x & \text{if } \vec{\ell} \text{ is empty} \\
    \lam{\rvec{i}}{\rvec{r}}{\rhs(x[\rvec{r}])}{\up_x(\vec{\ell}')[\rvec{i}]} & \text{if } \vec{\ell} = \mat{x[\rvec{r}]\\\vec{\ell}'} \\
    \up_x(\vec{\ell}') & \text{if } \vec{\ell} = \mat{y[\rvec{r}]\\\vec{\ell}'} \mathrlap{\text{ where } x \neq y}\hspace{2.35em}
  \end{cases}
\end{equation}
Here, the scalars $\rvec{i} \in \VV^{\arity(x)}$ do not occur freely in $\rvec{r}, \rhs(x[\rvec{r}])$, or $\up_x(\vec{\ell}')$.
Moreover, we define $\up \Def [x/\up_x(\vec{\ell}) \mid x \in \VV]$, i.e., $\up$ is the substitution that replaces each variable $x$ with the result of updating $x$ according to the body of the loop.
\begin{example}[$\up$]
  For \ref{leadLoop}, we have
  \[
    \up_i\mat{i\\a[i+1]\\a[i]} \quad = \quad \lambda.\ \mIte{[] =
      []}{i+1}{\up_i\mat{a[i+1]\\a[i]}} \quad = \quad i+1
  \]
  where $\lambda.\ \ldots$ denotes a function of arity $0$, and:
  \begin{align*}
    \up_a\mat{i\\a[i+1]\\a[i]} & {} = \up_a\mat{a[i+1]\\a[i]} =
   \lam{j}{i+1}{a[i]}{\up_a(a[i])[j]} \\ 
           & {} = \lam{j}{i+1}{a[i]}{\mIte{j=i}{a[i+1]}{\up_a()[j]}} \\
           & {} = \lam{j}{i+1}{a[i]}{\mIte{j=i}{a[i+1]}{a[j]}}
  \end{align*}
\end{example}
The definition of \eqref{eq:sem} corresponds to the \emph{read-over-write axioms}~\cite{mccarthy-arrays}:
\begin{align*}
  \rvec{i} = \rvec{r} \implies \up_x(x[\rvec{r}])[\rvec{i}] = \rhs(x[\rvec{r}]) \qquad\qquad
  \rvec{i} \neq \rvec{r} \implies \up_x(x[\rvec{r}])[\rvec{i}] = x[\rvec{i}]
\end{align*}
So our notion of arrays is compatible with the standard first-order theory of arrays.
However, our array expressions are more expressive than expressions in the standard
theory.
As an example, consider the array expression
\begin{equation}
  \label{eq:inexpressible}
  \lambda i.\ \mIte{0 \leq i \leq j}{0}{a[i]}.
\end{equation}
Here, the number of elements where $a$ and
\eqref{eq:inexpressible} differ is not fixed (it depends on $j$). Expressions for such arrays do not exist in the standard theory.

A \emph{state} $s$ maps each $x \in \VV$ to a function of arity $\arity(x)$.
The semantics $\sem{e}_s$ of an (array) expression $e$ w.r.t.\ a state $s$ is defined in
the obvious way
(see App.~\ref{sec:semantics}), and it is lifted to formulas over (array) expressions as usual.
In particular, the semantics of an array expression is a
function.
Given a first-order formula $\psi$, we often say that ``$\psi$ holds'' or
``$\psi$ is valid'',  which means that $\sem{\psi}_s = \true$ for all states $s$.
Now we can define the \emph{transition relation} of \ref{loop}.

\begin{definition}[Transition Relation]
  We define the relation ${\to_{\ref{loop}}}$ on states as follows:
  We have $s \to_{\ref{loop}} s'$ if $\sem{\phi}_s = \true$ and
  $s'(x) = \sem{\up(x)}_s$ for all $x \in \VV$.  
\end{definition}

\begin{example}[Transition Relation]
  Let $s(i) = 0$, $s(k) = 5$, and $s(a) = \lambda j.\ j$.
  Then $s \to_{\ref{leadLoop}} s'$ with $s'(i) =
  \sem{\up(i)}_s = \sem{i+1}_s =  s(i)+1 = 1$,
$s'(k) = s(k) = 5$,
  and:
  \begin{align*}
    s'(a) = \sem{\up(a)}_s = {} & \sem{\lambda j.\ \mIte{j = i+1}{a[i]}{\mIte{j=i}{a[i+1]}{a[j]}}}_s \\
    {} = {} & \lambda j.\ \mIte{j = 1}{0}{\mIte{j=0}{1}{j}}
  \end{align*}
\end{example}

\subsubsection{Acceleration:}

Our goal is to \emph{accelerate} loops like \ref{loop}.
\begin{definition}[Acceleration]
  \label{def:accel}
  Let ${\to^n_{\TT}}$ denote the $n$-fold closure of $\to_{\TT}$.
  An \emph{acceleration technique} is a partial function $\mathsf{accel}$, mapping loops
  to first-order formulas over $\VV$, $\VV' \Def \{x' \mid x \in \VV\}$, and a
  designated variable $n$ such that
  \[
     \mathsf{accel}(\TT)\;[x/s(x),x'/s'(x) \mid x \in \VV] \iff s \to^n_{\TT} s'
  \]
  holds for all $\TT \in \dom(\mathsf{accel})$, $n \in \NN_{>0}$, and all states $s,s'$.
  Here, $[x/s(x),x'/s'(x) \mid x \in \VV]$ is the substitution that replaces $x$ by $s(x)$
   and $x'$ by $s'(x)$, for all $x \in \VV$.
\end{definition}
So acceleration yields a formula over $\VV$ and $\VV'$, representing the values of the program variables before
and after $n$ loop iterations, and a dedicated loop counter $n$.
This formula characterizes the $n$-fold closure of \ref{loop} precisely.\footnote{Some acceleration techniques also construct under-approximations, but since we show how to accelerate certain loops precisely, \Cref{def:accel} does not allow approximations.}
\begin{example}[Acceleration]
  \label{ex:accel}
  Consider the loop
  \[
    \mWhile{i > 0}{
      i \assign i - 1
    }
  \]
  Here, an acceleration technique may yield the formula $i - n \geq 0 \land i'
  = i - n$.
\end{example}
As explained in \Cref{sec:overview},  many acceleration techniques rely on \emph{closed
form expressions} that characterize the values of the variables after $n$ \pagebreak[3] iterations of
the loop. \vspace*{-.3cm} 
\begin{definition}[Closed Forms]
  \label{def:closed-var}
  We call an expression $x^{(n)}$ over $\VV \cup \{n\}$ a \emph{closed form} for $x
  \in \VV$ if $x^{(n)} = \up^n(x)$ holds for all $n \in \NN$.
\end{definition}
For integers, closed forms are usually computed via \emph{recurrence solving}.
\begin{example}[Closed Forms]
  For \Cref{ex:accel}, a closed form for $i$ is computed by solving the recurrence equation $i^{(n)} = i^{(n-1)} - 1$ with the initial condition $i^{(0)} = i$.
\end{example}
Given closed forms for each $x \in \VV$, an acceleration technique can return
\[
  \left( \forall k \in [0 \twodots n).\ \phi[x/x^{(k)} \mid x \in \VV] \right) \land \bigwedge_{x \in \VV(\ref{loop})} x' = x^{(n)}.
\]
However, most acceleration techniques try to avoid quantifiers, as
they are challenging for automated reasoning.
Instead, the calculus from \cite{acceleration-calculus} can be used (see \Cref{sec:overview}), which
can often compute suitable quantifier-free formulas.
Therefore, in the sequel, we focus on computing closed forms for loops involving arrays.

\section{Closed Forms for Lvalues}
\label{sec:closed}

We now show how to compute closed forms for array loops.
To this end, we impose mild restrictions.
First, we require that the indices in $\vec{\ell}$ behave monotonically w.r.t.\ the
lexicographic order (where we have $\rvec{r} < \rvec{r}'$ if there is an $i \in [1
\twodots |\rvec{r}|]$ such
that $r_i < r'_i$ and $r_j = r'_j$ hold for all $j \in [1 \twodots i)$). 
\ref{loop} is
\emph{$x$-increasing} if the indices $\rvec{r}$ where $x$ is updated by the loop increase their values weakly in each iteration.

\begin{definition}[Monotonic Loop]
  \label{def:monotonic}
  \ref{loop} is \emph{$x$-increasing} or \emph{$x$-decreasing} if
  \[
    \bigwedge_{x[\rvec{r}] \in \vec{\ell}} \rvec{r} \leq \up(\rvec{r}) \qquad \text{or}
    \qquad \bigwedge_{x[\rvec{r}] \in \vec{\ell}} \rvec{r} \geq \up(\rvec{r})
  \]
  is valid, respectively.
  If \ref{loop} is $x$-increasing or $x$-decreasing, then it is \emph{$x$-mono\-ton\-ic}.
  If \ref{loop} is $x$-monotonic for each $x \in \VV$, then \ref{loop} is \emph{monotonic}.
\end{definition}
\begin{example}\label{ex:monotonic}
  Reconsider \ref{leadLoop}, where $a$ is accessed via the
  indices $i$ and $i+1$ in $\vec{\ell}$.
    Since we have $i \leq \up(i) = i+1$, \ref{leadLoop}
     is $a$-increasing and thus $a$-monotonic.
  As $i$ is a scalar variable, \ref{leadLoop} is also $i$-monotonic.
  Thus, \ref{leadLoop} is monotonic.
\end{example}
For simplicity, in the sequel we assume that $x$-monotonic loops are $x$-increasing.
Lifting our approach to arbitrary monotonic loops is straightforward.

A monotonic loop \ref{loop}
  is \emph{array-solvable}, or \emph{a-solvable} for short, if each set $\Lval(r_i)$ only refers to array cells that have been written in the previous iteration, or only to array cells that have not been written yet.

\begin{definition}[A-Solvable Loop]
  \label{def:solvable}
  Let $\LL$ be the smallest set such that:
  \[
    \Lval(\vec{r}) \subseteq \LL \qquad \text{and} \qquad \text{if $x[\rvec{r}] \in \LL$, then $\Lval(\rvec{r}) \subseteq \LL$}
  \]
  An lvalue $x[\rvec{r}] \in \LL$ is
  \begin{itemize}
  \item \emph{trivial} if for all $y[\rvec{r}'] \in
    \vec{\ell}$,
    we have $y \notin \VV(x[\rvec{r}])$
      \item \emph{inductive} if $x[\up(\rvec{r})] \in \vec{\ell}$
      \item \emph{displacing} if for all $\rvec{r}'$ such that $x[\rvec{r}'] \in \vec{\ell}$, \pagebreak[3]
        we have $\rvec{r}' < \up(\rvec{r})$
  \end{itemize}
  \ref{loop} is \emph{a-solvable} if it is monotonic, each element of $\LL$ is trivial,
  inductive, or displacing, and one of the following holds for each $i \in
  [1 \twodots |\vec{r}|]$:
  \begin{enumerate}[label=(\alph*)]
  \item \label{it:rec} all lvalues in $\Lval(r_i)$ are trivial or inductive
  \item \label{it:dis} all lvalues in $\Lval(r_i)$ are displacing
  \end{enumerate}
\end{definition}
So $x[\rvec{r}]$ is trivial if none of the variables occurring in $x[\rvec{r}]$ are changed by \ref{loop}.
Inductive lvalues $x[\rvec{r}]$ give rise to inductive definitions, as reading $x[\rvec{r}]$ in the $n^{th}$ iteration yields the value that
has been written to the same cell (i.e., to $x[\up(\rvec{r})]$) in the $(n-1)^{th}$ iteration.
In contrast, the cell that is referenced by a displacing lvalue $x[\rvec{r}]$ has not been written before it is being read, as all (increasing) indices that were used for write accesses to $x$ in earlier iterations are smaller than $\rvec{r}$ in the $n^{th}$ iteration.
Note that trivial lvalues are a special case of displacing lvalues.
For $x$-decreasing loops, the inequation $\rvec{r}' < \up(\rvec{r})$ becomes $\rvec{r}' > \up(\rvec{r})$.

\begin{example}[A-Solvable]
  For \ref{leadLoop}, we have  $r_1 = i+1$, $r_2 = a[i]$,
  $r_3=a[i+1]$, and $\LL = \{i, a[i], a[i+1]\}$.
 For $\Lval(r_1) = \{i\}$, we have $i \in \vec{\ell}$, i.e., $i$ is inductive.
 For $\Lval(r_2) = \{a[i]\}$, we have $a[\up(i)] = a[i+1] \in \vec{\ell}$, so
$a[i]$ is inductive, too.
 For $\Lval(r_3) = \{a[i+1]\}$, we have $a[\up(i+1)] = a[i+2]$.
So $a[i+1]$ is displacing, as  $i < i+2$ and $i+1 < i+2$ hold for the accesses $a[i]$ and $a[i+1]$ to $a$ in $\vec{\ell}$.
  Hence, \ref{leadLoop} is a-solvable.
\end{example}
Earlier acceleration techniques for arrays could essentially only handle
displacing lvalues without approximations, i.e., loops where each cell is
written at most once, see \Cref{sec:related}.
Moreover, they treated arrays and scalars differently.
Our new notion of inductive lvalues allows us to
also handle examples like \ref{leadLoop}, and to
unify the treatment of arrays and scalars (by regarding them as arrays of dimension $0$).
\Cref{ex:mix} shows why we disallow mixing inductive and displacing lvalues in $\Lval(r_i)$.

\begin{example}[Mixing Inductive and Displacing Lvalues]
  \label{ex:mix}
  Consider the loop
  \[
    \mWhile{i < k}{
      \mat{i \\ a[i + 1]} \assign \mat{i+1 \\ a[i] + a[i+1]}
    }
  \]
  where $a[i]$ is inductive and $a[i+1]$ is displacing.
  We have
  $\up^{n}(a[i]) = \sum_{j=i}^{i+n} a[j]$, which is not an expression as defined in \Cref{sec:preliminaries}.
For such loops, extensions of the underlying theory are required, see, e.g., \cite{array-sums}.
We leave that to future work.
\end{example}
We first show how to compute closed forms for all lvalues in $\LL$.
\begin{definition}[Closed Forms for Lvalues]
  \label{def:closed}
  We call an expression
  $\ell^{(n)}$ over $\VV \cup \{n\}$ a \emph{closed form} for $\ell \in \Lval$ if $\ell^{(n)} = \up^n(\ell)$ holds for all $n \in \NN$.
\end{definition}
For trivial lvalues, computing closed forms is straightforward.
\begin{lemma}[Closed Forms for Trivial Lvalues]
  \label{lem:trivial}
  If $\ell \in \LL$ is trivial, then $\ell^{(n)} \Def \ell$ is a closed form for $\ell$.
\end{lemma}
\Cref{thm:displacing} yields closed forms for displacing lvalues $x[\rvec{r}]$ if closed forms for all $\ell \in \Lval(\rvec{r})$ are known.
Here, we use the notation $\rvec{r}^{(n)} \Def \rvec{r}[\ell / \ell^{(n)} \mid \ell \in \Lval(\rvec{r})]$.\footnote{A similar observation has been used in \cite{fase09} to deduce invariants.}

\begin{restatable}[Closed Forms for Displacing Lvalues]{theorem}{displacing}
  \label{thm:displacing}
  If \ref{loop} is monotonic and $x[\rvec{r}] \in \LL$ is displacing, then
  $
    x[\rvec{r}]^{(n)} \Def x[\rvec{r}^{(n)}]
  $
  is a closed form for $x[\rvec{r}]$.
\end{restatable}\pagebreak[2]
\makeproof*{thm:displacing}{
  \displacing*
  \begin{proof}
    By \Cref{def:closed}, we need to show $x[\rvec{r}]^{(n)} = \up^n(x[\rvec{r}])$.
    We have:
    \usetagform{default}%
    \begin{align*}
      \up^n(x[\rvec{r}]) & {} = \up^n(x)[\up^n(\rvec{r})] \\
                         & {} = x[\up^n(\rvec{r})] \tag{by \Cref{lem:displacing}} \\
                         & {} = x[\rvec{r}^{(n)}] \tag{by \Cref{def:closed}} \\
                         & {} = x[\rvec{r}]^{(n)}
    \end{align*}
    \usetagform{box}%
    \qed
  \end{proof}  
}
\begin{example}
  \label{ex:displacing}
  Reconsider \ref{leadLoop}, where the only displacing lvalue is $a[i+1]$.
  Indeed,
  \[
    \up^n(a[i+1]) = a[(i+1)^{(n)}] = a[i+n+1]
  \]
  holds for all $n \in \NN$, i.e., the value that is referenced by the lvalue $a[i+1]$ after $n$ iterations is equal to the value that is referenced by $a[i+n+1]$ initially.
\end{example}
For inductive lvalues, a-solvable loops give rise to systems of recurrence equations.\footnote{Our definition of recurrence equations is less general than usual, as we only need recurrences of order $1$, and the initial conditions are always the same in our setting.}

\begin{definition}[Recurrence Equations]
  \label{def:rec-eqs}
  Let $\RR$ be a set of symbols with $n \notin \RR$.
  A \emph{recurrence equation} (over $\RR$) is an equation $\rec' = e$ where $\rec \in \RR$ and $e$ is an arithmetic expression over $\RR$.
  A substitution $\theta$ which maps symbols from $\RR$ to arithmetic
  expressions over $\RR \cup \{n\}$ is a \emph{solution} for $\rec'=e$ if
  \[
    \theta(\rec)[n/n+1] = \theta(e) \qquad \text{and} \qquad \theta(\rec)[n/0] = \rec
  \]
  are valid.
  A set of recurrence equations is also called a \emph{system}, and $\theta$ is a solution for such a system if it is a solution for each of its elements.
\end{definition}
\begin{example}
  Consider the recurrences
  $i' = i + 1$ and $j' = j + i$.
  Then $\theta$ with $\theta(i) = i + n$ and $\theta(j) = j + \tfrac{1}{2} \cdot n^2 + n \cdot i - \tfrac{1}{2} \cdot n$
  is a solution, as we have:
  \begin{align*}
    \theta(i)[n/n+1] & {} = i+n+1 = \theta(i)+1 = \theta(i+1) \qquad\qquad
    \theta(i)[n/0] = i+0 = i\\
    \theta(j)[n/n+1] & {} = j + \tfrac{1}{2} \cdot n^2+ \tfrac{1}{2} \cdot n + n \cdot i + i
                      = \theta(j) + n + i = \theta(j) + \theta(i) = \theta(j+i)\\
    \theta(j)[n/0] & {} = j + \tfrac{1}{2} \cdot 0^2 + 0 \cdot i - \tfrac{1}{2} \cdot 0 = j
  \end{align*}
\end{example}
To construct such systems, we apply the \emph{lvalue-substitution}
$\sigma_{\rec} \Def [\ell / \rec_{\ell} \mid \ell \in \Lval]$, which replaces every
top-level
occurrence of an lvalue $\ell$ by a fresh symbol $\rec_{\ell}$.
Now we define the recurrences that yield closed forms for inductive lvalues.

\begin{definition}[$\Rec$]
  \label{def:rec}
  $\Rec$ is the system of recurrence equations that contains $\rec_{x[\rvec{r}]}' = r\sigma_\rec$ where $r = \rhs(x[\up(\rvec{r})])$ for each inductive $x[\rvec{r}] \in \LL$.
\end{definition}
So $\Rec$ is a system of recurrences over $\{\rec_{\ell} \mid \ell \in \LL\}$.
As the lvalues in $\rvec{r}$ are updated simultaneously, the value of $x[\rvec{r}]$ in the \emph{next} iteration corresponds to the value that is written to $x[\up(\rvec{r})]$ in the current iteration, so we have $r = \rhs(x[\up(\rvec{r})])$.
This is in line with \Cref{def:solvable}, which ensures $x[\up(\rvec{r})] \in \vec{\ell} = \dom(\rhs)$.

\begin{example}[$\Rec$]
  \label{ex:rec}
  For \ref{leadLoop}, we get $\Rec =\{\rec'_i = \rec_i + 1, \rec_{a[i]}' = \rec_{a[i]}\}$.
  Note that we indeed have $\rhs(i) = i+1$ and
  $\rhs(a[\up(i)]) = \rhs(a[i+1]) = a[i]$.
  Then $\theta$ with $\theta(\rec_i) = \rec_i + n$ and $\theta(\rec_{a[i]}) = \rec_{a[i]}$ is a solution of $\Rec$, as:
  \begin{align*}
    \theta(\rec_i)[n/n+1] & {} = (\rec_i + n)[n/n+1] = \rec_i + n + 1 = \theta(\rec_i + 1) \\
    \theta(\rec_{a[i]})[n/n+1] & {} = \rec_{a[i]}[n/n+1] = \rec_{a[i]} = \theta(\rec_{a[i]}) \\
    \theta(\rec_i)[n/0] & {} = (\rec_i + n)[n/0] = \rec_i \quad\qquad \theta(\rec_{a[i]})[n/0] = \rec_{a[i]}[n/0] = \rec_{a[i]}
  \end{align*}
\end{example}
To solve $\Rec$, standard techniques for recurrence solving can be used (see,
e.g.~\cite{purrs}).
While these techniques are very powerful, there are of course cases where $\Rec$ cannot be solved, so that our approach cannot accelerate \ref{loop}.
In the sequel, we assume that $\Rec$ can be solved.
For inductive lvalues $\ell \in \LL$, the following theorem shows that a solution of $\Rec$ allows us to construct a closed form.
\begin{restatable}[Closed Forms for Inductive Lvalues]{theorem}{inductive}
  \label{thm:inductive}
  Let \ref{loop} be a-solvable and let $\theta$ be a solution for $\Rec$.
  Then for each inductive $\ell \in \LL$,
  $
    \ell^{(n)} \Def \theta(\rec_{\ell})\sigma^{-1}_\rec
  $
  is a closed form for $\ell$.
  Here, $\sigma^{-1}_\rec$ is the inverse of $\sigma_\rec$.
\end{restatable}
\makeproof*{thm:inductive}{
  \inductive*
  \begin{proof}
    We use induction on $n$ to show $\ell^{(n)} = \up^n(\ell)$.
    \usetagform{default}%
    \begin{description}
    \item[Induction Base -- $n = 0$:]
      \begin{align*}
        & \ell^{(n)}[n/0] \\
        {} = {} & \theta(\rec_{\ell})\sigma^{-1}_\rec[n/0] \tag{by def.\ of $\ell^{(n)}$} \\
        {} = {} & \theta(\rec_{\ell})[n/0]\sigma^{-1}_\rec \tag{as $n \notin \dom(\sigma^{-1}_\rec) \cup \VV(\img(\sigma^{-1}_\rec))$}\\
        {} = {} & \sigma^{-1}_\rec(\rec_{\ell}) \tag{by \Cref{def:rec-eqs}} \\
        {} = {} & \ell \tag{by def.\ of $\sigma_\rec$} \\
        {} = {} & \up^0(\ell) \\
        {} = {} & \up^n(\ell) \tag{as $n=0$}
      \end{align*}
    \item[Induction Step -- $n = m+1$:]
      \begin{align*}
        & \ell^{(n)}[n/m+1] \\
        {} = {} & x[\rvec{r}]^{(n)}[n/m+1] \tag{where $\ell = x[\rvec{r}]$}\\
        {} = {} & \theta(\rec_{x[\rvec{r}]})\sigma^{-1}_\rec[n/m+1] \tag{by def.\ of $\ell^{(n)}$} \\
        {} = {} & \theta(\rec_{x[\rvec{r}]})[n/m+1]\sigma^{-1}_\rec \tag{as $n \notin \dom(\sigma^{-1}_\rec) \cup \VV(\img(\sigma^{-1}_\rec))$}\\
        {} = {} & \theta(\rhs(x[\up(\rvec{r})])\sigma_\rec)[n/m]\sigma^{-1}_\rec \tag{by \Cref{def:rec-eqs,def:rec}}\\
        {} = {} & \theta(\rhs(x[\up(\rvec{r})])\sigma_\rec)\sigma^{-1}_\rec[n/m] \tag{as $n \notin \dom(\sigma^{-1}_\rec) \cup \VV(\img(\sigma^{-1}_\rec))$} \\
        {} = {} & \rhs(x[\up(\rvec{r})])[\ell'/\theta(\rec_{\ell'})\sigma^{-1}_\rec[n/m] \mid \ell' \in \Lval] \tag{by def.\ of $\sigma_\rec$} \\
        {} = {} & \rhs(x[\up(\rvec{r})])[\ell'/\up^m(\ell') \mid \ell' \in \Lval] \tag{IH} \\
        {} = {} & \up^m(\rhs(x[\up(\rvec{r})])) \\
        {} = {} & \up^{m}(\up(x)[\up(\rvec{r})]) \tag{by \eqref{eq:sem}} \\
        {} = {} & \up^{m}(\up(x[\rvec{r}])) \\
        {} = {} & \up^{m+1}(x[\rvec{r}]) \\
        {} = {} & \up^{m+1}(\ell) \tag{as $\ell = x[\rvec{r}]$} 
      \end{align*}
    \end{description}
    \usetagform{box}%
    \qed
  \end{proof}
}
\begin{example}
  \label{ex:inductive}
  Continuing \Cref{ex:rec}, we get:
  \begin{align*}
    i^{(n)} & {} = \theta(\rec_i)\sigma^{-1}_\rec = (\rec_i+n)\sigma^{-1}_\rec = i+n & \hspace{0.75em}
    a[i]^{(n)} & {} = \theta(\rec_{a[i]})\sigma^{-1}_\rec = \rec_{a[i]}\sigma^{-1}_\rec = a[i]
  \end{align*}
  Indeed, we have $i+n = \up^n(i)$ and $a[i] = \up^n(a[i])$ for all $n \in \NN$, i.e., the value that is referenced by $a[i]$ after $n$ iterations is equal to the initial value of $a[i]$.
\end{example}
\Cref{lem:trivial}, \Cref{thm:displacing}, and \Cref{thm:inductive} give rise to an
algorithm for computing closed forms for \emph{all} lvalues in $\LL$ (as defined in
\Cref{def:solvable}) whenever  \ref{loop} is a-solvable:
\begin{enumerate}
\item compute closed forms for all trivial lvalues via \Cref{lem:trivial}
\item compute closed forms for all inductive lvalues via \Cref{thm:inductive}
\item while there is an lvalue whose closed form has not yet been computed
  \begin{enumerate}
  \item \label{step:pick} pick an $x[\rvec{r}] \in \LL$ where the closed forms for all $\ell \in \Lval(\rvec{r})$ are known
  \item compute the closed form for $x[\rvec{r}]$ as in \Cref{thm:displacing}
  \end{enumerate}
\end{enumerate}
For Step (\ref{step:pick}), if $x[\rvec{r}] \in \LL$ is displacing and the closed form for $\ell \in \Lval(\rvec{r})$ has not yet been computed, then $\ell$ is displacing, too, and thus one can consider $\ell$ instead.
As $\ell$ is a subterm of $x[\rvec{r}]$, this yields a terminating approach for picking~$x[\rvec{r}]$.

\section{Closed Forms for Arrays}
\label{sect:qe}

We now show how to compute closed forms for arrays, given closed forms for all relevant lvalues, i.e., for all $\ell \in \LL$ (as defined in \Cref{def:solvable}).
The following theorem is an important step towards closed forms for
arrays.
\begin{restatable}[Towards Closed Forms for Arrays]{lemma}{arrays}
  \label{thm:closed}
  Let \ref{loop} be a-solvable and let $x \in \VV$.
  For each $x[\rvec{r}] \in \vec{\ell}$,
    we define the following auxiliary predicates:
  \begin{align*}
    \written_{\rvec{r}}(m,\rvec{c}) & {} \iff \rvec{r}^{(m-1)} = \rvec{c} \\
    \notwritten_x(m,n,\rvec{c}) & {} \iff \bigwedge_{x[\rvec{r}] \in \vec{\ell}} \forall m' \in [m \twodots n].\ \neg \written_{\rvec{r}}(m',\rvec{c}) \\
    \lastwrite_{x[\rvec{r}]}(m,n,\rvec{c}) & {} \iff \written_{\rvec{r}}(m,\rvec{c}) \land \notwritten_{x}(m+1,n,\rvec{c})
  \end{align*}
  Then the following holds for all $x[\rvec{r}] \in \vec{\ell}$ and all $\rvec{c} \in
  \ZZ^{\arity(x)}$:
  \begin{equation}
    \tag{\protect{\ensuremath{\textsc{Quant}}}}
    \label{eq:fo}
    \forall m \in [1 \twodots n].\ \lastwrite_{x[\rvec{r}]}(m,n,\rvec{c}) \implies \up^n(x)[\rvec{c}] = \rhs(x[\rvec{r}])^{(m-1)}
  \end{equation}
\end{restatable}
\makeproof*{thm:closed}{
  \arrays*
  \begin{proof}
    Let $\rvec{c} \in \ZZ^{\arity(x)}$ and $m \in [1 \twodots n]$ be arbitrary but fixed
    and assume that $\lastwrite_{x[\rvec{r}]}(m,n,\rvec{c})$ holds.
    Then $\notwritten_x(m+1,n,\rvec{c})$ implies
    \begin{equation}
      \label{eq:lastwrite}
      \hat{\rvec{r}}^{(m'-1)} \neq \rvec{c} \qquad \text{for all } m' \in (m,n]
\text{ and all }  x[\hat{\rvec{r}}] \in \vec{\ell}.
    \end{equation}
    and $\written_{\rvec{r}}(m,\rvec{c})$ implies $\rvec{r}^{(m-1)} = \rvec{c}$.
    Thus,
    \usetagform{default}%
    \begin{align*}
      & \up^n(x)[\rvec{c}] \\
      {} = {} & \up^{m}(x)[\rvec{c}] \tag{by \eqref{eq:lastwrite}} \\
      {} = {} & \left(\left(\up^{(m-1)} \circ [y/\up_y(\vec{\ell}) \mid y \in \VV]\right) (x)\right)[\rvec{c}] \tag{by def.\ of $\up$}\\
      {} = {} & \up^{(m-1)}(x[y/\up_y(\vec{\ell}) \mid y \in \VV])[\rvec{c}] \\
      {} = {} & \up^{(m-1)}(\up_x(\vec{\ell}))[\rvec{c}] \\
      {} = {} & \up^{(m-1)}(\lam{\rvec{i}}{\rvec{r}}{\rhs(x[\rvec{r}])}{\ldots})[\rvec{c}] \tag{$\dagger$}\\
      {} = {} & (\lam{\rvec{i}}{\up^{m-1}(\rvec{r})}{\up^{m-1}(\rhs(x[\rvec{r}]))}{\ldots})[\rvec{c}] \\
      {} = {} & (\lam{\rvec{i}}{\rvec{r}^{(m-1)}}{\rhs(x[\rvec{r}])^{(m-1)}}{\ldots})[\rvec{c}] \\
      {} = {} & (\lam{\rvec{i}}{\rvec{r}^{(m-1)}}{\rhs(x[\rvec{r}])^{(m-1)}}{\ldots})[\rvec{r}^{(m-1)}] \tag{as $\rvec{r}^{(m-1)} = \rvec{c}$} \\
      {} = {} &  \rhs(x[\rvec{r}])^{(m-1)}
    \end{align*}
    \usetagform{box}%
    For the step marked with $(\dagger)$, note that the order of the cases in the definition of $\up_x(\vec{\ell})$ is irrelevant due to \eqref{eq:distinct-indices}, and thus we may assume that the case $\rvec{i} = \rvec{r}$ is tested first without loss of generality.
    \qed
  \end{proof}
}

\noindent
Intuitively, $\written_{\rvec{r}}(m,\rvec{c})$ checks whether the write access $x[\rvec{r}] \gets \ldots$ modifies the cell $x[\rvec{c}]$ in the $m^{th}$ iteration.
The predicate $\notwritten_x(m,n,\rvec{c})$ is true if $x[\rvec{c}]$ does not
change
between the $m^{th}$ and the $n^{th}$ iteration.
The predicate $\lastwrite_{x[\rvec{r}]}(m,n,\rvec{c})$ holds if the last write access to
$x[\rvec{c}]$ (up until the $n^{th}$ iteration) takes place in the $m^{th}$ iteration via the assignment $x[\rvec{r}] \gets \ldots$.
Note that we have $x[\rvec{r}] \in \vec{\ell}$, so $\rhs(x[\rvec{r}])$ is well defined.
As the final value of $x[\rvec{c}]$ is written by $x[\rvec{r}] \gets \rhs(x[\rvec{r}])$ in the $m^{th}$ iteration, the final value of $x[\rvec{c}]$ is equal to the value of the right-hand side after the $(m-1)^{th}$ iteration, i.e., $\rhs(x[\rvec{r}])^{(m-1)}$.
\begin{example}[Towards Closed Forms for Arrays]
  \label{ex:qf}
  Continuing \Cref{ex:inductive}, we obtain the following definitions for $\written$ and
  $\notwritten$ due to the \pagebreak[3] lvalues $a[i]$ and $a[i+1]$:
  \begin{align*}
    \written_{i}(m,c) & {} \iff i + m - 1 = c  \qquad\qquad \written_{i+1}(m,c) \iff i + m = c \\
    \notwritten_a(m,n,c) & {} \iff \forall m' \in [m \twodots n].\ \neg\written_{i}(m',c) \land
    \neg\written_{i+1}(m',c) 
  \end{align*}
  Here, we combined two universal quantifiers to
  simplify $\notwritten$.
  The definition of\linebreak $\lastwrite$ is as in \Cref{thm:closed}.
  Then by \Cref{thm:closed},  for all $c \in \ZZ$ we have:
  \begin{equation}
    \label{eq:with-quantifiers}
    \begin{aligned}
      \forall m \in [1 \twodots n].\ \lastwrite_{a[i]}(m,n,c) & {} \implies \up^n(a)[c] =
      (a[i+1])^{(m-1)} = 
      a[i+m] \\
      \forall m \in [1 \twodots n].\ \lastwrite_{a[i+1]}(m,n,c) & {} \implies \up^n(a)[c]
      = (a[i])^{(m-1)}        
      = a[i] 
    \end{aligned}
  \end{equation}
\end{example}
\Cref{thm:closed} indicates that a closed form for $x$ might have the form
\begin{align*}
  x^{(n)} \Def \lambda \rvec{c}.\ & \mIf{1 \leq m \leq n \land \lastwrite_{x[\rvec{r}_1]}(m,n,\rvec{c})}{\rhs(x[\rvec{r}_1])^{(m-1)}}\ \mathbf{else\ if}\ \ldots \\
  \mElse\ &\mIf{1 \leq m \leq n \land\lastwrite_{x[\rvec{r}_{k}]}(m,n,\rvec{c})}{\rhs(x[\rvec{r}_{k}])^{(m-1)}} \ \mElse\ x[\rvec{c}]
\end{align*}
where the value of $m$ is uniquely determined by $\lastwrite$.
However, due to the additional free variable $m$, such an expression is not a closed form in the sense of \Cref{def:closed-var}, and thus it is not suitable for acceleration.
Moreover, $\lastwrite$ is defined in terms of $\notwritten$, and the latter contains universal quantifiers, which complicates subsequent automated reasoning.

Thus, we now show how to eliminate both the variable $m$ and the universal quantifiers in $\notwritten$ in the important special case that for every $x[\rvec{r}] \in \vec{\ell}$,
$\up(\rvec{r}) - \rvec{r} = \rvec{d}$ is a vector of
constants (which may differ for each $x[\rvec{r}] \in \vec{\ell}$).
Note that if $\rvec{d}$ is a vector of constants, then we have $\rvec{r}^{(n)} = \rvec{r} + \rvec{d} \cdot n$.

\subsubsection{Eliminating the Quantifier from $\notwritten$:}

For $\notwritten$, we get:
\begin{align*}
  & \forall m' \in [m \twodots n].\ \neg\written_{\rvec{r}}(m',\rvec{c}) \\
  {} \iff {} & \forall m' \in [m \twodots n].\ \rvec{r}^{(m'-1)} \neq \rvec{c} & \hspace{-6em}(\text{by def.\ of $\written$})\\
  {} \iff {} & \forall m' \in [m \twodots n].\ \rvec{r} + \rvec{d} \cdot (m'-1) \neq \rvec{c} & \hspace{-6em}(\text{as $\rvec{r}^{(n)} = \rvec{r} + \rvec{d} \cdot n$})\\
  {} \iff {} & \forall m' \in [m \twodots n].\ \bigvee_{\mathclap{i \in [1 \twodots |\rvec{r}|]}} r_i + d_i \cdot (m'-1) \neq c_i \\
  {} \iff {} &  \bigvee_{\mathclap{\substack{i \in [1 \twodots |\rvec{r}|]\\d_i \geq 0}}}
  c_i < r_i + d_i \cdot (m-1) \lor r_i + d_i \cdot (n-1) < c_i \quad \lor {} \\
  & \bigvee_{\mathclap{\substack{i \in [1 \twodots |\rvec{r}|]\\d_i<0}}} c_i < r_i + d_i
  \cdot (n-1) \lor r_i + d_i \cdot (m-1) < c_i \quad \lor \quad \bigvee_{\mathclap{\substack{i \in [1 \twodots |\rvec{r}|]\\|d_i| > 1}}} d_i \notdivides (c_i-r_i) \tag{qe} \label{eq:qe}
\end{align*}
Again, ``$\divides$'' denotes divisibility, so $d_i \notdivides (c_i-r_i)$ is true iff $d_i$ does not divide $c_i-r_i$.
In the last step, the first two sub-formulas $\bigvee \ldots$ cover the case that
\begin{equation}
  \label{eq:oob}
  c_i \notin \left[\min_{m'\in[m,n]}(r_i+d_i\cdot(m'-1)),\max_{m'\in[m,n]}(r_i+d_i\cdot(m'-1))\right],
\end{equation}
which immediately implies $r_i + d_i \cdot (m'-1) \neq c_i$.
Here, we consider the cases $d_i \geq 0$ (first $\bigvee \ldots$) and $d_i<0$ (second $\bigvee \ldots$) separately.
If \eqref{eq:oob} does not hold, then we have $\forall m' \in [m \twodots n].\ r_i+d_i\cdot(m'-1) \neq c_i$ iff $c_i-r_i$ is not divisible by $d_i$.

\subsubsection{Eliminating the Quantifier from \normalfont{\eqref{eq:fo}}:}

Next, we eliminate $m$ from the negated matrix of \eqref{eq:fo}, i.e., from
\begin{equation}
  \label{eq:negquant}
  \tag{\protect{\ensuremath{!\textsc{Quant}}}}
  m \in [1 \twodots n] \land \lastwrite_{x[\rvec{r}]}(m,n,\rvec{c}) \land \up^n(x)[\rvec{c}] \neq \rhs(x[\rvec{r}])^{(m-1)},
\end{equation}
by deriving an instantiation $e_{x[\rvec{r}]}$ so that
$\exists m.\, \eqref{eq:negquant} \iff \eqref{eq:negquant}[m/e_{x[\rvec{r}]}]$ and thus $\eqref{eq:fo} \iff \neg\,\eqref{eq:negquant}[m/e_{x[\rvec{r}]}]$ holds.
We get:
\usetagform{default}%
\begin{align*}
  \lastwrite_{x[\rvec{r}]}(m,n,\rvec{c}) \iff {} & \written_{\rvec{r}}(m,\rvec{c}) \land \ldots \tag{by def.\ of $\lastwrite$}\\
  {} \iff {} & \rvec{r}^{(m-1)} = \rvec{c} \land \ldots \tag{by def.\ of $\written$} \\
  {} \iff {} & \rvec{r} + \rvec{d} \cdot (m-1) = \rvec{c} \land \ldots \tag{by def.\ of $\rvec{r}^{(n)}$} \\
  {} \iff {} & \bigwedge_{\mathclap{\substack{i \in [1 \twodots |\rvec{r}|]\\d_i
  = 0}}} r_i = c_i \land \bigwedge_{\mathclap{\substack{i \in [1 \twodots |\rvec{r}|]\\d_i
  \neq 0}}} \; d_i {\divides} (c_i-r_i) \land m = (c_i-r_i) \div d_i + 1 \land \ldots
\end{align*}
\usetagform{box}%
The last step results from solving the equation $\rvec{r} + \rvec{d} \cdot (m-1) = \rvec{c}$ for $m$.
More explicitly, this equation means $\bigwedge_{i=1}^{|\rvec{r}|} r_i + d_i \cdot (m-1) =
c_i$. So for each conjunct, we obtain $m = \frac{c_i-r_i}{d_i}+1$ if $d_i \neq 0$, or $r_i=c_i$ if $d_i=0$.
As $m$ is an integer, the former is equivalent to 
  $m = (c_i-r_i) \div d_i + 1 \land d_i {\divides} (c_i-r_i)$.
So if there is a $d_i \neq 0$, then $m$ can be eliminated by propagating the equality $m=(c_i-r_i) \div d_i + 1$.

If \emph{all} elements of $\rvec{d}$ are $0$, then $\rvec{r}$ does not change at all,  and we obtain:
\usetagform{default}%
\begin{align*}
  \lastwrite_{x[\rvec{r}]}(m,n,\rvec{c}) \iff {} & \written_{\rvec{r}}(m,\rvec{c}) \land \notwritten_x(m+1,n,\rvec{c}) \tag{def.\ of $\lastwrite$}\\
  {} \iff {} & \written_{\rvec{r}}(m,\rvec{c}) \land m \geq n
\end{align*}
\usetagform{box}%
The reason for the last step is that, in this case, the write access $x[\rvec{r}] \gets \ldots$ overwrites the same cell in every iteration.
So $\written_{\rvec{r}}(m,\rvec{c})$ implies $\written_{\rvec{r}}(m+1,\rvec{c})$, which contradicts $\notwritten_x(m+1,n,\rvec{c})$, unless we have $m+1>n$.
In the latter case, $\notwritten_x(m+1,n,\rvec{c})$ is trivially true.
As $m$ ranges over $[1,n]$ in \eqref{eq:negquant}, $m + 1 > n$ implies $m=n$, so we can
instantiate $m$ with $n$ if~$\rvec{d} = [0,\ldots,0]$.

Hence, if the indices of all array accesses only change by constants in every iteration, then all quantifiers that occur in \Cref{thm:closed} can be eliminated.

\begin{example}[Quantifier Elimination]
  \label{ex:qe}
  For \ref{leadLoop}, we have $\rvec{d} = 1$ for both $a[i]$ and $a[i+1]$.
  As explained above, we can simplify $\notwritten_a(m,n,c)$ from \Cref{ex:qf}:
  \begin{align}
    \label{eq:unchanged-qe}
    (c < i + m - 1 \lor i + n - 1 < c) \land (c < i + m \lor i + n < c)
  \end{align}
  Moreover, the final formulas \eqref{eq:with-quantifiers} from \Cref{ex:qf} can be simplified to:
  \begin{equation}
    \label{eq:without-quantifiers}
    \begin{aligned}
      1 \leq c-i+1 \leq n \land \lastwrite_{a[i]}(c-i+1,n,c) & {} \implies \up^n(a)[c] = a[c+1] \\
      1 \leq c-i \leq n \land \lastwrite_{a[i+1]}(c-i,n,c) & {} \implies \up^n(a)[c] = a[i] 
    \end{aligned}
  \end{equation}
  Here, we propagated the equality $m=c-i+1$ (implied by $\lastwrite_{a[i]}(m,n,c)$) in the first line, and the equality $m=c-i$ (implied by $\lastwrite_{a[i+1]}(m,n,c)$) in the second line.
  For the first occurrence of $\lastwrite$ in \eqref{eq:without-quantifiers}, we get:
  \usetagform{default}%
  \begin{align*}
    &\lastwrite_{a[i]}(c-i+1,n,c) \\
    {} \iff {} & i^{(c-i)}=c \land \notwritten_a(c-i+2,n,c) \tag{by def.\ of $\lastwrite$ and $\written$} \\[-0.5em]
    {} \iff {} & i+c-i=c \land \notwritten_a(c-i+2,n,c) \iff \notwritten_a(c-i+2,n,c) \overset{\scriptsizeeqref{eq:unchanged-qe}}{\iff} \true
  \end{align*}
  \usetagform{box}%
  Similarly, for the second occurrence of $\lastwrite$ in \eqref{eq:without-quantifiers}, we get:
  \vspace{-0.3em}
  \begin{align*}
    \lastwrite_{a[i+1]}(c-i,n,c) \iff \notwritten_a(c-i+1,n,c) \overset{\scriptsizeeqref{eq:unchanged-qe}}{\iff} i + n - 1 < c
  \end{align*}
  Thus, \eqref{eq:without-quantifiers} simplifies to:
  \begin{equation}
    \label{eq:closed-lambda}
    \begin{aligned}
      0 \leq c-i < n \implies \up^n(a)[c] = a[c+1] \qquad\quad 1 \leq c-i = n \implies \up^n(a)[c] = a[i]
    \end{aligned}
  \end{equation}
  The left implication states that $a[i+1],\ldots,a[i+n]$ are shifted to the left
  by one position, and the right implication states that $a[i]$ gets moved to index $i+n$.
\end{example}

\subsubsection{Putting it all together:}

Now we can finally compute closed forms for arrays.
Recall that the quantified variable $m$ from \eqref{eq:fo} was eliminated via instantiation.
For each $x[\rvec{r}] \in \vec{\ell}$, let $e_{x[\rvec{r}]}$ be the corresponding instantiation of $m$, i.e., the expression such that
\begin{align*}
  \forall m \in [1 \twodots n].\ \lastwrite_{x[\rvec{r}]}(m,n,\rvec{c}) & {} \implies \up^n(x)[\rvec{c}] = \rhs(x[\rvec{r}])^{(m-1)} \\
  \text{iff} \quad 1 \leq e_{x[\rvec{r}]} \leq n \land \lastwrite_{x[\rvec{r}]}(e_{x[\rvec{r}]},n,\rvec{c}) & {} \implies \up^n(x)[\rvec{c}] = \rhs(x[\rvec{r}])^{(e_{x[\rvec{r}]}-1)}.
\end{align*}
Then we obtain our main theorem:
\begin{restatable}[Closed Forms for Arrays]{theorem}{arraythm}
  \label{thm:closed-arrays}
  Let $x \in \VV$,
  let \ref{loop} be a-solvable, assume that $\Rec$ can be solved, and let $x[\rvec{r}_1],\ldots,x[\rvec{r}_k]$ be all elements of $\vec{\ell}$ of the form $x[\ldots]$.
  If $\up(\rvec{r}_i) - \rvec{r}_i$ is a vector of constants for every $i \in [1 \twodots k]$, then
  \begin{equation}
    \hspace*{-.2cm}\begin{aligned}
      &x^{(n)} \Def \lambda \rvec{c}. \\
      &\phantom{\mElse}\ \mIf{1 \leq e_{x[\rvec{r}_1]} \leq n \land \lastwrite_{x[\rvec{r}_1]}(e_{x[\rvec{r}_1]},n,\rvec{c})}{\rhs(x[\rvec{r}_1])^{(e_{x[\rvec{r}_1]}-1)}} && \!\!\!\!\mathbf{else\ if}\ldots \\
      &\mElse\ \mIf{1 \leq e_{x[\rvec{r}_k]} \leq n \land\lastwrite_{x[\rvec{r}_{k}]}(e_{x[\rvec{r}_k]},n,\rvec{c})}{\rhs(x[\rvec{r}_{k}])^{(e_{x[\rvec{r}_k]}-1)}} && \!\!\!\!\mElse\ x[\rvec{c}]
    \end{aligned}
  \end{equation}
  is a closed form for $x$.
\end{restatable}
\makeproof*{thm:closed-arrays}{
  \arraythm*
  \begin{proof}
    Let $\rvec{c} \in \ZZ^{\arity(x)}$ be arbitrary but fixed.
    If there is an $i \in [1 \twodots k]$ such that $1 \leq e_{x[\rvec{r}_i]} \leq n \land
    \lastwrite_{x[\rvec{r}_{i}]}(e_{x[\rvec{r}_i]},n,\rvec{c})$, then $x^{(n)}[\rvec{c}] =
    \up^n(x)[\rvec{c}]$ follows from \Cref{thm:closed} and correctness of our quantifier
    elimination technique. 
    Otherwise, by correctness of our quantifier elimination technique, we have
    \[
      \forall m \in [1 \twodots n].\ \neg \lastwrite_{x[\rvec{r}_i]}(m,n,\rvec{c})
    \]
    for all $i \in [1 \twodots k]$.
    By definition of $\lastwrite$, this implies
    \[
      \forall m \in [1 \twodots n].\ \rvec{r}_i^{(m-1)} \neq \rvec{c}.
    \]
    Thus, the cell $x[\rvec{c}]$ is not written by the loop, i.e., we have $\up^n(x)[\rvec{c}] = x[\rvec{c}] = x^{(n)}[\rvec{c}]$. \qed
  \end{proof}
}
So the closed form $x^{(n)}$ is the array that maps every index $\rvec{c}$ to the updated
right-hand side $\rhs(x[\rvec{r}_i])^{(e_{x[\rvec{r}_i]}-1)}$ (which is constructed using
the closed forms for lvalues from \Cref{sec:closed}),
where the last write access to the cell $x[\rvec{c}]$ took place in the
${e_{x[\rvec{r}_i]}}^{th}$ iteration via an assignment to $x[\rvec{r}_i]$.
As shown when deriving \eqref{eq:qe}, all quantifiers can be eliminated from the definition of $\notwritten$ (and thus $\lastwrite$), so \Cref{thm:closed-arrays} gives rise to \emph{quantifier-free} closed forms for arrays.
\begin{example}
  In \Cref{ex:qe}, we instantiated $m$ with $e_{a[i]} \Def c-i+1$ and $e_{a[i+1]} \Def c-i$, respectively.
  Thus, similar to \eqref{eq:without-quantifiers}, we obtain the closed form
  \[
    \begin{aligned}
    a^{(n)} = \lambda c.\ & \mIf{1 \leq c-i+1 \leq n \land \lastwrite_{a[i]}(c-i+1,n,c)}{a[c+1]} && \mElse\\
    & \mIf{1 \leq c-i \leq n \land \lastwrite_{a[i+1]}(c-i,n,c)}{a[i]} && \mElse\ a[c]
    \end{aligned}
    \]
since $(a[i+1])^{(c - i + 1 - 1)} = a[c+1]$ and $(a[i])^{(c - i - 1)} = a[i]$. By eliminating all remaining quantifiers and applying simplifications, analogously to \eqref{eq:closed-lambda}, we obtain:
\begin{align}
  \label{eq:finalClosedForm}
    a^{(n)} = \lambda c.\ \mIf{0 \leq c-i < n}{a[c+1]}\ \mElse\ \mIf{1 \leq c-i = n}{a[i]}\ \mElse\ a[c]
  \end{align}
\end{example}
At this point, we obtained closed form expressions in the sense of \Cref{def:closed-var}, i.e., they are suitable for acceleration.
As they are quantifier-free, they are also suitable for automation via SMT by handling
$\lambda$s with \emph{lemmas on demand}, see App.\ \ref{sec:lambdas}.

\section{Related Work}
\label{sec:related}

Acceleration techniques for arrays
were already studied in \cite{aligators,underapprox15,alberti13}, where \cite{aligators} is\linebreak restricted to loops where arrays are partitioned into read- and write-only arrays.

The technique from \cite{underapprox15} does not require such a partitioning, but if the
same array cell is written more than once (as in \ref{leadLoop}, where $a[i] \gets \ldots$
overwrites a cell that has
\pagebreak[3]
already been written by $a[i+1] \gets \ldots$ in the previous iteration), then \cite{underapprox15} approximates.
Moreover, to compute a closed form for an array $a$, \cite{underapprox15} requires that the closed forms for right-hand sides of updates $a[i] \gets \ldots$ are already known.
This means that $a$ itself must not occur on the right-hand side.

In \cite{alberti13}, it is shown how to accelerate \emph{local ground assignments}.
Their definition imposes certain restrictions that are not required in our setting.
In particular, the loop body must only update each array $a$ at a single index $i$.
So their approach does not apply to, e.g., \ref{leadLoop}.
In contrast to \cite{underapprox15},
restricted forms of occurrences of 
the array $a$ are possible on the right-hand side of $a[i] \gets \ldots$

A similar approach to our elimination of the quantifiers from $\notwritten$ in \Cref{sect:qe} is also used in \cite{alberti13}, but our technique for eliminating the quantifier from \eqref{eq:fo} is specific to our approach.
Moreover, \cite{alberti13} also uses $\lambda$s for the result of acceleration, but only to simplify the presentation.
In their implementation, they use quantifiers instead, whereas we use SMT solving with $\lambda$s.

Subsequently, the results from \cite{alberti13} have been used in \cite{albertiJAR15} to show decidability of safety for a certain class of array programs.
In an orthogonal line of research, \cite{alberti15} describes a framework to combine acceleration techniques.

\section{Implementation and Evaluation}
\label{sec:eval}

We implemented our approach in \tool{LoAT}~\cite{loat}, where we use it in our model checking algorithm ABMC~\cite{abmc}, which integrates acceleration into Bounded Model Checking (BMC).
ABMC lifts acceleration from single-path loops to systems with complex control flow, i.e., it also supports multi-path loops, nested loops, etc.
\tool{LoAT} uses our SMT solver \tool{SwInE}~\cite{swine}, which is built on top of \tool{Z3}~\cite{z3}.
The input format of \tool{LoAT} are Constrained Horn Clauses (CHCs)~\cite{chc-comp-format}, where (un)safe verification tasks correspond to (un)satisfiable CHCs.

We used our novel tool {\sf HornKlaus}~\cite{hornklaus} to transform the {\sf C}
benchmarks from the category {\tt ReachSafety-Arrays} of the Software Verification
Competition (SV-COMP \cite{SV-COMP}) to CHCs.
{\sf HornKlaus} does not (yet) aim to preserve (un)safety in all cases, as it does not take all intricacies of the {\sf C} standard into account.
Instead, its purpose is to generate CHCs that mirror the given \emph{algorithm} as closely as possible.
Thus, {\sf HornKlaus} cannot yet be used
to compare
  CHC solvers with {\sf C} verifiers, but it has certain advantages over other tools that convert {\sf C} to CHCs:
{\sf SeaHorn}~\cite{seahorn} converts multi-dimensional arrays to one-dimensional arrays, resulting in a larger gap between the {\sf C} program and the CHCs, and {\sf Korn}~\cite{korn} often yields non-linear CHCs (which are not yet supported by {\sf LoAT}) when {\sf HornKlaus} does not.

Currently, {\sf HornKlaus} supports a fragment of {\sf C} that corresponds to the CHCs
handled by {\sf LoAT} (e.g., it does not support pointers, structs, or bit-wise operations).
Thus, it can only transform 201 of the 439 SV-COMP benchmarks into CHCs, resulting in our first collection of benchmarks, called $\mathtt{SV}$.
However, most of these benchmarks are satisfiable, but ABMC's main purpose is proving $\unsat$.
To obtain unsatisfiable instances, we modified {\sf HornKlaus} so that it negates the
assertions that characterize the error states,
resulting in our second set of benchmarks $\mathtt{SV}_\neg$.

We compared {\sf LoAT} with the state-of-the-art CHC solvers {\sf Eldarica}
\cite{eldarica}, {\sf Golem}\linebreak \cite{GolemCAV23},
and {\sf Z3} \cite{z3}.
We tested several configurations, and used the most powerful ones for proving $\unsat$:
The default configuration of {\sf Eldarica~2.2.1}, the symbolic execution engine of {\sf Golem~0.9.0}, and the BMC implementation of {\sf Z3~4.15.4}.
Additionally, we compared with a \underline{b}ase\underline{l}ine variant {\sf LoAT BL} without the contributions from \Cref{sec:closed,sect:qe}, i.e., it can only accelerate loops without arrays.

Note that some SV-COMP benchmarks contain non-linear arithmetic, which is not supported by {\sf Golem}, and sometimes prevents {\sf Eldarica} from computing interpolants.
Moreover, each of our benchmark collections contains $13$ examples with ``curried'' arrays (see \Cref{sec:preliminaries}), which are not supported by {\sf LoAT}.

We could not compare with \tool{MCMT}~\cite{alberti13,alberti15,albertiJAR15}, a
state-of-the-art model checker that supports array acceleration.
The reason is that \tool{MCMT}'s array acceleration introduces quantifiers, which are then handled via over-approximations (see \cite{mcmt-manual} for details), and thus they cannot be used for proving unsafety.
This shows the importance of our quantifier elimination technique from \Cref{sect:qe} and the use of $\lambda$s, as these two ingredients allow us to use acceleration for proving unsatisfiability.

We used \href{https://help.itc.rwth-aachen.de/service/rhr4fjjutttf/article/fbd107191cf14c4b8307f44f545cf68a/}{CLAIX-2023-HPC nodes} of the \href{https://help.itc.rwth-aachen.de/service/rhr4fjjutttf/}{RWTH Uni\-ver\-si\-ty High Performance Computing Cluster} with 10560 MiB memory limit and 60~s timeout per example.

\begin{figure}[t]
  \hspace{0.05\textwidth}
  \begin{minipage}{0.4\textwidth}
    \begin{tabular}{|l||c||c||c|c|}
      \hline \multicolumn{2}{|c||}{}                                       & \checkmark & unsat & sat \\
      \hline\hline \multirow{4}{*}{\ \rotatebox{90}{$\mathtt{SV}$ (201 ex.) }}        & \tool{LoAT}     & 25         & 24    & 1   \\
      \hhline{~----}                                     & \tool{Eldarica} & 2          & 0     & 2   \\
      \hhline{~----}                                     & \tool{LoAT BL}     & 2          & 2     & 0   \\
      \hhline{~----}                                     & \tool{Z3}       & 1          & 0     & 1   \\
      \hhline{~----}                                     & \tool{Golem}    & 0          & 0     & 0   \\
      \hline\hline \multirow{4}{*}{\  \rotatebox{90}{$\mathtt{SV}_\neg\!$ (201 ex.)  }} & \tool{LoAT}     & 129        & 120   & 9   \\
      \hhline{~----}                                     & \tool{Eldarica} & 78         & 68    & 10  \\
      \hhline{~----}                                     & \tool{LoAT BL} & 65    & 65    & 0  \\
      \hhline{~----}                                     & \tool{Z3}       & 62         & 62    & 0   \\
      \hhline{~----}                                     & \tool{Golem}    & 30         & 30    & 0   \\
      \hline
    \end{tabular}
    \label{tab}
  \end{minipage}
  \begin{minipage}{0.5\textwidth}
    \begin{tikzpicture}[scale=0.73]
      \begin{axis}[
        legend pos=south east,
        ylabel=solved instances from $\mathtt{SV}_\neg$,
        y label style={at={(axis description cs:0.1,.5)},anchor=south},
        y tick label style={rotate=90},
        xticklabel={$\pgfmathprintnumber{\tick}$s},
        ymax=135,
        xmin=-5,
        xmax=60,
        legend columns=2]
        \addplot[color=black,solid,thick] table[col sep=comma,header=false,x index=0,y index=1] {abmc_inv.csv}; \addlegendentry{\tool{\scriptsize LoAT}}
        \addplot[color=red,dashed,thick] table[col sep=comma,header=false,x index=0,y
          index=1] {eld_inv.csv}; \addlegendentry{\tool{\scriptsize Eldarica}}
        \addplot[color=OliveGreen,dotted,thick] table[col sep=comma,header=false,x index=0,y index=1] {abmc_no_arrays_inv.csv}; \addlegendentry{\tool{\scriptsize LoAT BL}}
        \addplot[color=violet,densely dashed,thick] table[col sep=comma,header=false,x index=0,y index=1] {z3_inv.csv}; \addlegendentry{\tool{\scriptsize Z3}}
        \addplot[color=blue,densely dotted,thick] table[col sep=comma,header=false,x index=0,y index=1] {golem_inv.csv}; \addlegendentry{\tool{\scriptsize Golem}}
      \end{axis}
    \end{tikzpicture}
  \end{minipage}
  \vspace*{.1cm}
\end{figure}

The results of our experiments are shown in the table above.
On $\mathtt{SV}$, all tools perform quite badly, which confirms
the observation from \Cref{sec:intro} that model checking array programs
is notoriously difficult.
\tool{LoAT} is the only tool that can prove\linebreak $\unsat$ in some cases.
On $\mathtt{SV}_\neg$, all tools can prove $\unsat$ sometimes, but \tool{LoAT} is significantly ahead of its closest competitors.
Whenever another tool can prove unsatisfiability, \tool{LoAT} can do
so as well (apart from $3$ of the examples with ``curried'' arrays).
The fact that \tool{LoAT} outperforms \tool{LoAT BL} shows that its good results are indeed due to our novel contributions.
For $\mathtt{SV}_\neg$, the plot on the right depicts the runtimes for solved instances.
It shows that \tool{LoAT} solves most instances quickly, i.e., our approach is not only competitive in terms of solved instances, but it is also quite efficient.
See \cite{loat-web,website} for more details on \tool{LoAT} and our evaluation.

\vspace*{-.2cm}

\subsubsection{Conclusion:}

We presented a new loop acceleration technique for array loops.
By using the novel notion of \emph{inductive lvalues},
it can
accelerate more classes of loops than earlier techniques, and it handles
arrays and scalars uniformly.
Due to the use of $\lambda$-expressions and our technique for eliminating
quantifiers from \Cref{sect:qe}, it
  improves
over previous SMT-based model checking algorithms for array programs,
in particular for proving unsafety, as demonstrated by our evaluation.

\vspace*{-.2cm}

\subsubsection{Acknowledgments:} We thank Matthias Heizmann
for vital initial discussions.

\begin{credits}
  \subsubsection{Disclosure of Interests:}
  The authors have no competing interests to declare that are relevant to the content of
  this article.
\end{credits}

\bibliographystyle{splncs04}
\bibliography{refs,crossrefs,strings}

\begin{appendix}
  \section{Reasoning about Lambdas}
\label{sec:lambdas}

Using the approach from \Cref{sect:qe}, we obtain quantifier-free closed forms for arrays, which allows us to encode many verification problems as SMT problems.

\begin{example}[Encoding Verification as SMT]
  Consider our running example \ref{leadLoop} once more, and assume that we want to verify the Hoare triple
  \[
    \{\phi \land i < k\}\ref{leadLoop}\{\psi\},
  \]
  where $\phi$ and $\psi$ are quantifier-free first-order formulas over $\VV$.
  Together with the closed form \eqref{eq:finalClosedForm}, we can use an acceleration
  techniques like \cite{acceleration-calculus} to derive the following precise
  characterization of the transitive closure of \ref{leadLoop}:
  \[
 n > 0 \land  i+n \leq k \land i' = i + n \land a' = a^{(n)}
  \]
  Then $\{\phi \land i < k\}\ref{leadLoop}\{\psi\}$ is a valid Hoare triple iff
  \begin{equation}
    \label{eq:smt}
    \phi \land i < k \land n > 0 \land i+n \leq k \land i' = i + n  \land a' = a^{(n)} \land i' \geq k \land \neg\psi[a/a',i/i']
  \end{equation}
  is unsatisfiable.
\end{example}

However, while all major SMT solvers support $\lambda$s as input (via non-standard extensions of SMT-LIB or their API), in our experience, they can rarely solve instances where $\lambda$-expressions cannot be eliminated via $\beta$-reduction, i.e., they can hardly handle cases where $\lambda$s cannot be evaluated immediately.
However, even in the simple example above, this is not the case, as the $\lambda$-expression $a^{(n)}$ occurs in a position that is not $\beta$-reducible (in the equation $a' = a^{(n)}$).

Thus, we now describe a simple and incomplete theory solver to deal with $\lambda$s, which turns out to be very effective in our setting.
More precisely, we assume that the input to our theory solver is a conjunction $\phi$ of literals that are either \emph{arithmetic literals} that relate arithmetic expressions (via the usual arithmetic relations, including (in)divisibility), or \emph{array literals} that relate array expressions via ${=}$ or ${\neq}$.
For array literals, we additionally require that the expressions on both sides have the same arity.
Here, it is important to note that, like SMT-LIB, we consider \emph{extensional} arrays, i.e., we have
\begin{equation}
  \tag{\sc Ext}
  \label{eq:ext}
  a = a' \iff \forall \rvec{i} \in \ZZ^{\arity(a)}.\ a[\rvec{i}] = a'[\rvec{i}].
\end{equation}
Then our goal is to either produce a model, or to return $\unsat$ if no such model exists.

The first ingredient of our theory solver is \emph{equality propagation}.
\begin{example}[Equality Propagation]
  To see why equality propagation is useful, assume that we want to check the Hoare triple
  \begin{equation}
    \label{eq:hoare-triple}
    \{b=a \land j=i<k\}\ref{leadLoop}\{a[i] = b[j]\}.
  \end{equation}
  So $b$ and $j$ store copies of the initial values of $a$ and $i$, respectively,
  and the post-condition claims that the final value of $a[i]$ is equal to the initial value of $a[i]$.
  Instantiating \eqref{eq:smt} correspondingly yields
  \[
    b=a \land j=i<k \land n > 0 \land i+n \leq k \land i' = i + n \land a' = a^{(n)} \land
    i' \geq k \land a'[i'] \neq b[j]. 
  \] 
  After propagating the equality $a' = a^{(n)}$, we obtain
  \[
    b=a \land j=i<k \land n > 0 \land i+n \leq k \land i' = i + n \land i' \geq k \land a^{(n)}[i'] \neq b[j].
  \]
  Now $a^{(n)}[i']$ can be $\beta$-reduced, and the resulting $\lambda$-free formula can be handled by off-the-shelf SMT techniques.
  Indeed, by taking the equality $i' = i + n$ into account, it becomes apparent that $a^{(n)}[i']$ evaluates to $a[i]$ due to the second case of \eqref{eq:finalClosedForm} if $n>0$, or due to the third case if $n=0$, i.e., we obtain
  \[
    b=a \land j=i \land \ldots \land a[i] \neq b[j],
  \]
  which is clearly unsatisfiable.
  Thus, \eqref{eq:hoare-triple} is valid.
\end{example}

\begin{algorithm}[t]
  \KwIn{conjunction of literals $\phi$}
  \KwResult{a model of $\phi$, $\unsat$, or $\unknown$}
  $\phi \gets \phi[(p \neq q) / (p[\rvec{i}] \neq q[\rvec{i}]) \mid (p \neq q) \in \phi, \arity(p) > 0, \rvec{i} \subseteq \VV \text{ fresh}]$ \tcp*{remove $\neq$}
  \Repeat{$\phi$ does not change anymore}{
    \lIf{$\phi$ contains a literal $x = p$ where $x \in \VV$}{$\phi \gets \phi[x/p]$ \tcp*[f]{propagate eq.}}
    \lIf{$\phi$ contains $(\lambda \rvec{i}.\ e)[\rvec{e}]$}{$\phi \gets \phi[(\lambda \rvec{i}.\ e)[\rvec{e}]/e[\rvec{i}/\rvec{e}]]$ \tcp*[f]{beta reduction}}
  }
  $\phi_{abs} \gets \phi[(\lambda \rvec{i}.\ e) / x_{\rvec{i}, e} \mid (\lambda
    \rvec{i}.\ e) \text{ occurs in } \phi, x_{\rvec{i}, e} \in \VV \text{ fresh}]$\;\label{line:abs}
  $\mathit{Idx} \gets \{\rvec{e} \mid a[\rvec{e}] \in \Lval(\phi), \arity(a) > 0 \}$\; \label{line:idx}
  \Repeat{$\phi_{abs}$ does not change anymore}{
    \lIf{$\phi_{abs}$ is $\unsat$}{\Return $\unsat$ \label{line:unsat}}
    let $\mathcal{M}$ be a model of $\phi_{abs}$\;
    \lIf{$\mathcal{M}$ is a model of $\phi$}{\Return $\mathcal{M}$ \label{line:sat}}
    \ForEach{$(p = q) \in \phi \text{ and } \rvec{e} \in \mathit{Idx} \text{ such that } \mathcal{M} \text{ violates } p[\rvec{e}] = q[\rvec{e}]$ \label{line:refine}}{$\phi_{abs} \gets \phi_{abs} \land p[\rvec{e}] = q[\rvec{e}]$ \tcp*{refinement} \label{line:refine-body}}
  }
  \Return $\unknown$ \tcp*[f]{refinement failed} \label{line:fail}
  \caption{Theory Solving with $\lambda$s}
  \label{alg}
\end{algorithm}

Due to \eqref{eq:ext}, an array literal $p \neq q$ can easily be converted into the
arithmetic literal $p[\rvec{i}] \neq q[\rvec{i}]$ where $\rvec{i}$ consists of fresh
(implicitly existentially quantified) scalar variables.
Using these two ingredients and $\beta$-reduction, our theory solver tries to eliminate as many $\lambda$-expressions as possible, which constitutes the first part of our theory solver, see \Cref{alg}.
Here, we slightly abuse notation by using ``substitutions'' that replace expressions or
literals instead of variables, and we do not mention arity constraints explicitly (i.e.,
whenever we write, e.g., $p[\rvec{i}]$, we implicitly assume $|\rvec{i}| = \arity(p)$).

Then we \emph{abstract} the SMT problem by replacing all remaining $\lambda$-expressions with fresh variables of the corresponding arity (Line~\ref{line:abs}).
If the resulting SMT problem is unsatisfiable, then the original problem $\phi$ is unsatisfiable as well, so our solver returns $\unsat$ (Line~\ref{line:unsat}).
Otherwise, we check if the model that we obtained for the abstraction is also a model for
$\phi$, and in that case, our solver returns it (Line~\ref{line:sat}).
If this is not the case, then we refine the abstraction by considering the set
$\mathit{Idx}$ of
all vectors of
expressions $\rvec{e}$ that occur syntactically as indices\footnote{This idea stems from the decision procedure for the \emph{array property fragment} \cite{array-property-fragment}.} in parts of $\phi$ that are not
below a $\lambda$, and all array literals $p = q$ that were abstracted (Lines~\ref{line:idx} and \ref{line:refine} -- note that $\beta$-redexes and array-inequations can \emph{always} be eliminated):
For each pair of a literal $p = q$ and an index $\rvec{e}$ of the corresponding arity
where $p[\rvec{e}] = q[\rvec{e}]$ is violated, we add $p[\rvec{e}] = q[\rvec{e}]$ to the
abstraction (Line~\ref{line:refine-body}).
If we find at least one such pair, then the refinement was successful and we search for a model for the refined abstraction.
Otherwise, the refinement fails and our solver returns $\unknown$
(Line~\ref{line:fail}).

To compute $\mathit{Idx}$ in Line~\ref{line:idx},
we lift $\Lval$ to expressions as follows:
\[
  \Lval(e) \Def \begin{cases}
    \{e\} & \text{if } e \in \Lval \\
    \Lval(e_1) \cup \Lval(e_2) & \text{if } e = e_1 \circ e_2\\
    &\text{\phantom{if} where $\circ$ is an arithmetic operator}\\
    \Lval(\mu) \cup \Lval(e_1) \cup \Lval(e_2) & \text{if } e = \mIte{\mu}{e_1}{e_2} \\
    \emptyset & \text{otherwise}
  \end{cases}
\]
Moreover, for quantifier-free formulas, we define:
\[
  \Lval(\mu) \Def \begin{cases}
    \Lval(\mu_1) \cup \Lval(\mu_2) & \text{if } \mu = \mu_1 \bullet \mu_2 \text{ where  $\bullet \in \{{\land},{\lor},\ldots\}$}\\
    & \text{\phantom{if} is a binary Boolean connective}\\
    \Lval(\mu') & \text{if } \mu = \neg\mu' \\
    \Lval(e_1) \cup \Lval(e_2) & \text{if } \mu = e_1 \bowtie e_2 \text{ is an arithmetic
      literal}\\
    \emptyset & \text{otherwise}
  \end{cases}
  \]
  So in particular, for array literals we have $\Lval(p=q) = \emptyset$, i.e., we do not
  consider lvalues below $\lambda$s.

  To see why our refinement in Line~\ref{line:refine-body}
  may fail and \Cref{alg} may reach
Line~\ref{line:fail}, consider the formula
\[
  (\lambda i.\ 0) = (\lambda i.\ 1)
\]
which gets abstracted to
\[
  x_{i,0} = x_{i,1},
\]
and assume that we obtain a model $\mathcal{M}$ with $\mathcal{M}(x_{i,0}) =
\mathcal{M}(x_{i,1}) =
\lambda i.\ 42$.
Then $\mathcal{M}$ is not a model for the original formula, but as $\mathit{Idx}$ is
empty, the refinement fails and \Cref{alg} returns $\unknown$.

Note that our approach is very similar to \emph{lemmas on demand}~\cite{lemmas-on-demand,lemmas-phd,DBLP:conf/fmcad/PreinerNB13}.
However, their technique is more sophisticated then our approach.
For example, our selection of the relevant indices $\mathit{Idx}$ is a lot simpler.
Moreover, for each array literal $a=b$, \emph{lemmas on demand} adds the implication
\[
  a \neq b \implies \exists \rvec{i}.\ a[\rvec{i}] \neq b[\rvec{i}]
\]
to the SMT problem \cite{lemmas-phd}, which is required for concluding satisfiability in cases where the refinement fails.
Thus, we return $\unknown$ in Line~\ref{line:fail}, but their algorithm is a decision procedure (assuming that the abstraction in Line~\ref{line:abs} yields formulas in a decidable logic).
However, to the best of our knowledge, \emph{lemmas on demand} has only been implemented in \tool{Bitwuzla}~\cite{bitwuzla}, which does not support integer arithmetic.
Thus, we use our own, incomplete implementation.
Improving it is a subject of future work.

  \appendixproofsection{Proofs}\label{sec:proofs}

We first prove the following auxiliary lemma, which shows that the value that is referenced by a displacing lvalue in the $n^{th}$ iteration has not been changed in previous iterations.
\begin{restatable}[Stability of Displacing Lvalues]{lemma}{stability}
  \label{lem:displacing}
  If \ref{loop} is monotonic and $x[\rvec{r}] \in \LL$ is displacing, then we have
  \[
    x[\up^n(\rvec{r})] = \up^m(x)[\up^n(\rvec{r})]
  \]
  for all $m \leq n$.
\end{restatable}
\begin{proof}
  We use induction on $m$, where the case $m = 0$ is trivial.
  Let $m = k+1$.
  Then for all $x[\rvec{r}'] \in \vec{\ell}$, we have:
  \usetagform{default}%
  \begin{align*}
    \up^k(\rvec{r}') & {} < \up^{k+1}(\rvec{r}) \tag{as $x[\rvec{r}]$ is displacing} \\
                         & {} \leq \up^{n}(\rvec{r}) \tag{by monotonicity, as $k+1 \leq n$}
  \end{align*}
  \usetagform{box}%
  Thus,
  \begin{equation}
    \label{eq:lem-displacing}
    \up^k(\rvec{r}') < \up^{n}(\rvec{r}) \qquad \text{ whenever } x[\rvec{r}'] \in \vec{\ell}
  \end{equation}
  Hence, we get
  \usetagform{default}%
  \begin{align*}
    \up^m(x)[\up^n(\rvec{r})] & {} = \up^{k+1}(x)[\up^n(\rvec{r})] \tag{as $m = k+1$} \\
                              & {} = \up^k(x)[\up^n(\rvec{r})] \tag{by \eqref{eq:sem} due to \eqref{eq:lem-displacing}} \\
                              & {} = x[\up^n(\rvec{r})] \tag{IH}
  \end{align*}
  \usetagform{box}%
  \qed
\end{proof}

\appendixproof*{thm:displacing}
\appendixproof*{thm:inductive}
\appendixproof*{thm:closed}
\appendixproof*{thm:closed-arrays}

  \section{Semantics of Expressions}
\label{sec:semantics}

In this section, we define the semantics of expressions as introduced in \Cref{sec:preliminaries}.

\begin{definition}[Semantics of Expressions]
  \label{def:semantics}
  Let $s$ with $\dom(s) \subseteq \VV$ be a function that maps variables $x \in \dom(s)$ to functions of
  arity $\arity(x)$. 
  The \emph{interpretation} $\sem{e}_s$ of an (array) expression $e$ w.r.t.\  $s$ is defined as follows:
  \[
    \sem{e}_s \Def
    \begin{cases}
      e & \text{if } e \in \ZZ \\
      \sem{e_1}_s \circ \sem{e_2}_s & \text{if } e = e_1 \circ e_2\\
      \sem{p}_s(\sem{\rvec{e}}_s) & \text{if } e = p[\rvec{e}], p \text{ is an array expression} \\
      s(e) & \text{if } e \in \VV \cap \dom(s) \\
      e & \text{if } e \in \VV \setminus \dom(s) \\
      \lambda \rvec{i}.\ \sem{e'}_{s|_{\dom(s) \setminus \rvec{i}}} & \text{if } e = \lambda \rvec{i}.\ e' \\
      \sem{e_1}_s & \text{if } e = \mIte{\mu}{e_1}{e_2} \text{ and } \sem{\mu}_s = \true \\
      \sem{e_2}_s & \text{if } e = \mIte{\mu}{e_1}{e_2} \text{ and } \sem{\mu}_s = \false
  \end{cases}
\]
    Here, we lift $\sem{\_}_s$ to vectors of expressions and formulas in the obvious way,
    and $s|_{\dom(s) \setminus \rvec{i}}$ denotes the restriction of $s$ to the domain
    $\dom(s) \setminus \rvec{i}$.
\end{definition}

\end{appendix}

\end{document}